\def\ifundefined#1{\expandafter\ifx\csname#1\endcsname\relax}
\let\arXiv = 1

\ifundefined{arXiv}
    \documentclass[review, number, nonatbib]{elsarticle}
    \biboptions{sort&compress}
\else
    \documentclass[10pt]{article}
\fi

\usepackage{xcolor} % Package to handle colors

\usepackage{cite} % Cites and compresses citation numbers

\usepackage{amsmath}
\usepackage{amssymb}
\usepackage{amsfonts}

\usepackage{algorithmic}

\usepackage{array} % Enhances array, tabular environments (strongly encouraged)

\usepackage[caption=false,font=footnotesize]{subfig}
\usepackage{stfloats} % Double column floats placed at the bottom
\usepackage{url} % Handling and breaking URLs

\usepackage{textcomp} % Allows for specific font issues as bullets and others

\makeatletter
    \@ifpackageloaded{xcolor}{}{\usepackage{xcolor}}
\makeatother
\newcommand{\uppercaseGreek}[1]{%
    \begingroup
    \ucmathlist\MakeUppercase{#1}%
    \endgroup
}

\newcommand{\ucmathlist}{%
    \def\alpha{\mathrm{A}}%
    \def\beta{\mathrm{B}}%
    \let\gamma=\Gamma
    \let\delta=\Delta
    \def\epsilon{\mathrm{E}}%
    \def\varepsilon{\mathrm{E}}%
    \def\zeta{\mathrm{Z}}%
    \def\eta{\mathrm{H}}%
    \let\theta=\Theta
    \let\vartheta=\Theta
    \def\iota{\mathrm{I}}%
    \def\kappa{\mathrm{K}}%
    \let\lambda=\Lambda
    \def\mu{\mathrm{M}}%
    \def\nu{\mathrm{N}}%
    \let\xi=\Xi
    \let\pi=\Pi
    \let\varpi=\Pi
    \def\rho{\mathrm{P}}%
    \def\varrho{\mathrm{P}}%
    \let\sigma=\Sigma
    \def\tau{\mathrm{T}}%
    \let\upsilon=\Upsilon
    \let\phi=\Phi
    \let\varphi=\Phi
    \def\chi{\mathrm{X}}%
    \let\psi=\Psi
    \let\omega=\Omega
}

\makeatletter
    \@ifpackageloaded{amsthm}{}{

        \usepackage{amsthm}
    }
\makeatother
\theoremstyle{plain}
    \newtheorem{theorem}{Theorem}
    \newtheorem{proposition}[theorem]{Proposition}
    \newtheorem{lemma}[theorem]{Lemma}
    \newtheorem{corollary}[theorem]{Corollary}
\theoremstyle{definition}
    \newtheorem{definition}{Definition}

\makeatletter
\def\renewtheorem#1{%
    \expandafter\let\csname#1\endcsname\relax
    \expandafter\let\csname c@#1\endcsname\relax
    \gdef\renewtheorem@envname{#1}
    \renewtheorem@secpar
}
\def\renewtheorem@secpar{\@ifnextchar[{\renewtheorem@numberedlike}{\renewtheorem@nonumberedlike}}
\def\renewtheorem@numberedlike[#1]#2{\newtheorem{\renewtheorem@envname}[#1]{#2}}
\def\renewtheorem@nonumberedlike#1{  
    \def\renewtheorem@caption{#1}
    \edef\renewtheorem@nowithin{\noexpand\newtheorem{\renewtheorem@envname}{\renewtheorem@caption}}
    \renewtheorem@thirdpar
}
\def\renewtheorem@thirdpar{\@ifnextchar[{\renewtheorem@within}{\renewtheorem@nowithin}}
\def\renewtheorem@within[#1]{\renewtheorem@nowithin[#1]}
\makeatother
\input{auxFiles/symbolDef.sty}
\input{auxFiles/berkeleyColors.sty}

\ifundefined{arXiv}
    \usepackage{graphicx}
\else
    \usepackage{cite}
    \usepackage[pdftex]{graphicx}
    \usepackage{accents}
    \usepackage{geometry}[margin = 1in]
    \usepackage{setspace}
    \newcommand\blfootnote[1]{%
        \begingroup
        \renewcommand\thefootnote{}\footnote{#1}%
        \addtocounter{footnote}{-1}%
        \endgroup
    }
\fi
\graphicspath{{../figures/}}
\DeclareGraphicsExtensions{.pdf} %,.jpeg,.png}
    \usepackage{IEEEtrantools}

\makeatletter
\def\munderbar#1{\underline{\sbox\tw@{$#1$}\dp\tw@\z@\box\tw@}}
\makeatother

\def\mthbA{\mathbf{\hat{\munderbar{\mathbf{A}}}}}
\def\mthbB{\mathbf{\hat{\munderbar{\mathbf{B}}}}}

\begin{document}

%%%%%%%%%%%%%%%%%%%%%%%%%%%%%%%%%%%%%%%%%%%%%%%%%%%%%%%%%%%%%%%%%%%%%%%%%%%%%%%%
%%%%                          DOCUMENT INFORMATION                          %%%%
%%%%%%%%%%%%%%%%%%%%%%%%%%%%%%%%%%%%%%%%%%%%%%%%%%%%%%%%%%%%%%%%%%%%%%%%%%%%%%%%

%%%%%%%%%%%%%%%%%%%%%%%%%%%%%%%%%%%%%%%%
%%%%          Paper Title           %%%%
%%%%%%%%%%%%%%%%%%%%%%%%%%%%%%%%%%%%%%%%

\def\theTitle{Distributed Linear-Quadratic Control with\\Graph Neural Networks}

\ifundefined{arXiv}
    \title{\theTitle\tnoteref{t1,t2}}
    \tnotetext[t1]{This work is supported by grants from ONR, NSF and AFOSR.}
    \tnotetext[t2]{Partial results have appeared in \cite{Gama21-LQRGNN}.}
\else
    \title{\theTitle}
\fi

%%%%%%%%%%%%%%%%%%%%%%%%%%%%%%%%%%%%%%%%
%%%%            Authors             %%%%
%%%%%%%%%%%%%%%%%%%%%%%%%%%%%%%%%%%%%%%%

\ifundefined{arXiv}
    \author[1]{Fer\hspace{0.015cm}nando Gama\corref{cor1}}
    \ead{fgama@rice.edu}
    \author[2]{Somayeh Sojoudi}
    \ead{sojoudi@berkeley.edu}
    \cortext[cor1]{Corresponding author}
    \address[1]{Department of Electrical and Computer Engineering, Rice University, Houston, TX, 77005 USA}
    \address[2]{Department of Electrical Engineering and Computer Sciences, University of California, Berkeley, CA, 94709 USA}
\else
    \author{Fer\hspace{0.015cm}nando Gama and Somayeh Sojoudi}
    \date{}
    \maketitle
\fi

%%%%%%%%%%%%%%%%%%%%%%%%%%%%%%%%%%%%%%%%%%%%%%%%%%%%%%%%%%%%%%%%%%%%%%%%%%%%%%%%
%%%%                                ABSTRACT                                %%%%
%%%%%%%%%%%%%%%%%%%%%%%%%%%%%%%%%%%%%%%%%%%%%%%%%%%%%%%%%%%%%%%%%%%%%%%%%%%%%%%%

\begin{abstract}
% 12 lines
Controlling network systems has become a problem of paramount importance. In this paper, we consider a distributed linear-quadratic problem and propose the use of graph neural networks (GNNs) to parametrize and design a distributed controller for network systems. GNNs exhibit many desirable properties, such as being naturally distributed and scalable. We cast the distributed linear-quadratic problem as a self-supervised learning problem, which is then used to train the GNN-based controllers. We also obtain sufficient conditions for the resulting closed-loop system to be input-state stable, and derive an upper bound on how much the trajectory deviates from the nominal value when the matrices that describe the system are not accurately known. We run extensive simulations to study the performance of GNN-based distributed controllers and show that they are computationally efficient and scalable.
% 30.030: Nonlinear Filtering
% 30: Nonlinear Signal Processing
% 80.020: Sensor Networks
% 80: Signal Processing Applications
\end{abstract}

\ifundefined{arXiv}
    \begin{keyword}
        distributed control, graph neural networks, graph signal processing
    \end{keyword}
    \maketitle
\else
    \blfootnote{This work is supported by grants from ONR, NSF and AFOSR.}
    \blfootnote{F. Gama is with the Department of Electrical and Computer Engineering, Rice University, Houston, TX 77005 USA email: fgama@rice.edu.}
    \blfootnote{S. Sojoudi is with the Department of Electrical Engineering and Computer Sciences, University of California, Berkeley, CA, 94709 USA e-mail: sojoudi@berkeley.edu.}
    \blfootnote{Partial results have appeared in \cite{Gama21-LQRGNN}.}
\fi

%%%%%%%%%%%%%%%%%%%%%%%%%%%%%%%%%%%%%%%%%%%%%%%%%%%%%%%%%%%%%%%%%%%%%%%%%%%%%%%%
%%%%                                                                        %%%%
%%%%                              INTRODUCTION                              %%%%
%%%%                                                                        %%%%
%%%%%%%%%%%%%%%%%%%%%%%%%%%%%%%%%%%%%%%%%%%%%%%%%%%%%%%%%%%%%%%%%%%%%%%%%%%%%%%%

\section{Introduction} \label{sec:intro}

%!TEX root = 00-distributedLQR.tex

%%%%%%%%%%%%%%%%%%%%%%%%%%%%%%%%%%%%%%%%%%%%%%%%%%%%%%%%%%%%%%%%%%%%%%%%%%%%%%%%
%%%%                                                                        %%%%
%%%%                              INTRODUCTION                              %%%%
%%%%                                                                        %%%%
%%%%%%%%%%%%%%%%%%%%%%%%%%%%%%%%%%%%%%%%%%%%%%%%%%%%%%%%%%%%%%%%%%%%%%%%%%%%%%%%
%%%% sec:intro
%%%%%%%%%%%%%%

The use of linear models to describe dynamical systems has found widespread use in many areas of physics, mathematics, engineering and economics \cite{Kailath80-LinearSystems}. Linear systems are mathematically tractable and can thus be used to derive properties, draw insights, and improve on our ability to successfully control these systems. In particular, designing optimal controllers that can steer the system into a desired state while minimizing some given cost has become a problem of paramount importance \cite{AndersonMoore89-LQR}.

Obtaining an optimal controller that minimizes a quadratic cost on the states and the actions, following a linear dynamic model, gives rise to the well-studied linear-quadratic control problem \cite{Dean20-SampleLQR}. As it happens, the optimal linear-quadratic controller is \emph{linear} and acts on the knowledge of the system state at a given time to produce the optimal control action for that time instant. Furthermore, when considering an infinite-time horizon for minimizing the quadratic cost, the resulting optimal controller is not only linear but also static, meaning that the same linear mapping is used between state and control action for all time instants.

Network systems are one particular class of dynamical systems that has become increasingly relevant. These systems are comprised of a set of interconnected components that are capable of exchanging information. They are further equipped with the ability to autonomously decide on an action to take based on the individual state of each component and the information relied through the communications with other neighboring components. The objective of controlling network systems is to coordinate the individual actions of the components so that they are conducive to the accomplishment of some global task \cite{Fattahi19-SingleSample}.

The dynamics of some network systems can be effectively described by a linear model. Thus, if such systems are coupled with a quadratic cost, a corresponding linear-quadratic problem is obtained. As such, the optimal control actions are readily available. While the optimal controllers are linear, they require information from the components in the network irrespective of their interconnections. That is, to compute the optimal controller, an additional unit capable of accessing all components instantaneously is required. In the context of network systems, this is called a \emph{centralized} approach.

Centralized controllers face limitations in terms of scalability and implementation. For increasingly large networks, the centralized unit requires more direct connections to all the components of the system. Similarly, the computational cost increases directly with the size of the network, since a single unit is responsible for computing the control actions of all the components. It is also less robust to changes in the network. A failed connection between the centralized unit and any of the components would render that component uncontrollable. %Likewise, if the centralized unit goes down, the entire system becomes unmanageable.

Network systems are characterized by the connections between components, which naturally impose a distributed structure on the flow of data. Fundamentally, it can be leveraged in the design of controllers. By doing so, one can overcome many of the limitations of centralized controllers. Thus, we focus on leveraging the data structure to obtain \emph{distributed} controllers. These are control actions that depend only on local information provided by components that share a connection and that can be computed separately by each component.

Imposing a distributed constraint on the linear-quadratic control problem renders it intractable in the most general case \cite{Witsenhausen68-Counterexample}. While there is a large class of distributed control problems that admit a convex formulation \cite{Rotkowitz06-DecentralizedConvex}, many of them lead to complex solutions that do not scale with the size of the network \cite{Fattahi19-TransformCentralized}. An alternative approach is to adopt a linear parametrization of the controller and find a surrogate of the original problem that admits a scalable solution. The resulting controller is thus a sub-optimal linear distributed controller, and stability and robustness analyses are provided \cite{Fazelnia17-LowerBoundLQR, Matni19-SystemLevelApproach, Fattahi19-LQR}.

However, even in the context of linear network systems with a quadratic cost, the optimal distributed controller may not be linear \cite{Witsenhausen68-Counterexample}. In this paper, we thus adopt a nonlinear parametrization of the controller. More specifically, we focus on the use of graph neural networks (GNNs) \cite{Gama20-GNNs}. GNNs consist of a cascade of blocks (commonly known as layers) each of which applies a bank of graph filters followed by a pointwise nonlinearity. GNNs exhibit several desirable properties in the context of distributed control. Most importantly, they are naturally local and distributed, meaning that by adopting a GNN as a mapping between states and actions, a distributed controller is automatically obtained. Furthermore, they are permutation equivariant and Lipschitz continuous to changes in the network \cite{Gama20-Stability}. These two properties allow them to scale up and transfer \cite{Ruiz20-Transferability}.

Distributed controllers leveraging neural network techniques can be found in \cite{Gama21-LQRGNN, Capella03-DistributedNN, Huang05-LargeScaleDecentralized, Srinivasan06-ContinuousOnlineLearning, Chen13-DecentralizedPID, Liu15-ContinuousTimeUnknown, Yang17-FormationControlRBF, Wang20-NeuroOptimal, Wang20-Wastewater, Gama21-ControlGNN}. These controllers typically use a distinct multi-layer perceptron (MLP) to parametrize the controller at each component \cite{Capella03-DistributedNN, Huang05-LargeScaleDecentralized, Srinivasan06-ContinuousOnlineLearning, Chen13-DecentralizedPID, Liu15-ContinuousTimeUnknown, Yang17-FormationControlRBF} or rely on adaptive critic control \cite{Wang20-NeuroOptimal, Wang20-Wastewater}. Assigning a separate MLP to each component implies that the number of parameters to learn increases proportionally with the size of the network system, becoming increasingly harder to train, and thus this approach is not scalable. The use of GNNs imposes a weight-sharing scheme that avoids scalability problems. These are leveraged in \cite{Gama21-ControlGNN} in the context of specific robotics problems. The distributed linear-quadratic problem using GNNs was investigated in our conference paper \cite{Gama21-LQRGNN}.

In this work, we focus on finding distributed controllers for the distributed linear-quadratic problem. Our main contributions are:

\begin{list}{}{
        \setlength{\labelwidth}{30pt}
        \setlength{\leftmargin}{30pt}
        \setlength{\labelsep}{10pt}
        \setlength{\itemsep}{5pt}
        \setlength{\topsep}{5pt}
        \setlength{\parskip}{0pt}
    }

    \item[\textbf{(C1)}] We propose to parametrize the distributed controller with a GNN, obtaining a naturally distributed architecture that is capable of capturing nonlinear relationships between input and output, as it was initially investigated in our preliminary work \cite{Gama21-LQRGNN}.

    \item[\textbf{(C2)}] We obtain an improved sufficient condition for closed-loop input-state stability of the controller.

    \item[\textbf{(C3)}] We study the problem of systems whose linear description is not accurately known. We analyze how the stability of the system changes and obtain an upper bound on the deviation of the trajectory from its nominal value.

    \item[\textbf{(C4)}] We present new simulations that provide better insight into GNN-based controllers for a distributed LQR problem.

\end{list}

The remainder of this paper is organized as follows. We formulate the linear-quadratic problem in Section~\ref{sec:LQR} and postulate the use of graph neural networks in Section~\ref{sec:GNN} as a practically useful nonlinear parametrization of the unknown distributed controller. We cast the distributed linear-quadratic problem as a self-supervised learning problem, which can be efficiently solved by traditional machine learning techniques. To study the effect of adopting a GNN-based controller on the entire dynamical system, we obtain a sufficient condition for the resulting closed-loop system to be input-state stable and derive an upper bound on the trajectory deviation from its nominal value when the system matrices are unknown and only estimates are available. We include numerical simulations in Section~\ref{sec:sims} to investigate the performance of GNN-based distributed controllers and their dependence on design hyperparameters, as well as their scalability. Conclusions are drawn in Section~\ref{sec:conclusions}. Proofs are provided in the appendix.

%%%%%%%%%%%%%%%%%%%%%%%%%%%%%%%%%%%%%%%%%%%%%%%%%%%%%%%%%%%%%%%%%%%%%%%%%%%%%%%%
%%%%                                                                        %%%%
%%%%                              LQR PROBLEM                               %%%%
%%%%                                                                        %%%%
%%%%%%%%%%%%%%%%%%%%%%%%%%%%%%%%%%%%%%%%%%%%%%%%%%%%%%%%%%%%%%%%%%%%%%%%%%%%%%%%

\section{The Linear-Quadratic Problem} \label{sec:LQR}

%!TEX root = 00-distributedLQR.tex

%%%%%%%%%%%%%%%%%%%%%%%%%%%%%%%%%%%%%%%%%%%%%%%%%%%%%%%%%%%%%%%%%%%%%%%%%%%%%%%%
%%%%                                                                        %%%%
%%%%                              LQR PROBLEM                               %%%%
%%%%                                                                        %%%%
%%%%%%%%%%%%%%%%%%%%%%%%%%%%%%%%%%%%%%%%%%%%%%%%%%%%%%%%%%%%%%%%%%%%%%%%%%%%%%%%
%%%% sec:LQR
%%%%%%%%%%%%

The linear-quadratic problem is one of the fundamental problems in optimal control theory \cite{AndersonMoore89-LQR}. Consider a system described by a state vector $\vcx(t) \in \fdR^{F}$ and controlled by an action $\vcu(t) \in \fdR^{G}$ at time $t \in \{0,1,2,\ldots\}$. The system evolves following a linear dynamic
% eq:linearDynamicSingle
\begin{equation} \label{eq:linearDynamicSingle}
    \vcx(t+1) = \mtbA \vcx(t) + \mtbB \vcu(t)
\end{equation}
determined by $\mtbA \in \fdR^{F \times F}$ called the \emph{system matrix} and $\mtbB \in \fdR^{F \times G}$ called the \emph{control matrix}. These two matrices are considered to be known and given in the problem formulation. The objective is to drive the system towards a desired, target state value. To this end, a \emph{controller} $\fnPhi: \fdR^{F} \to \fdR^{G}$ that maps the current state of the system $\vcx(t)$ into an appropriate action $\vcu(t) = \fnPhi(\vcx(t))$ is typically designed. In optimal control, it is desirable to find a controller that minimizes a given cost. In particular, the focus here is on the quadratic cost given by
% eq:quadraticCostSingle
\begin{equation} \label{eq:quadraticCostSingle}
    \fnJ \Big( \{\vcx(t)\} ,\{\vcu(t)\} \Big) = \sum_{t=0}^{\infty} \big(\vcx(t)^{\Tr} \mtbQ \vcx(t) + \vcu(t)^{\Tr} \mtbR \vcu(t) \big)
\end{equation}
for two known matrices $\mtbQ \in \fdR^{F \times F}$ and $\mtbR \in \fdR^{G \times G}$ such that $\mtbQ \succeq 0$ and $\mtbR \succ 0$, given in the problem formulation.

The linear-quadratic problem can be formulated as
% eq:LQRsingle, {eq:LQRsingleObj, eq:LQRsingleDynamics}
\begin{IEEEeqnarray}{CLR} \IEEEyesnumber \label{eq:LQRsingle}
    \min_{\fnPhi \in \fdPhi} & \fnJ \Big( \{\vcx(t)\} ,\{\vcu(t)\} \Big)
        \IEEEyessubnumber \label{subeq:LQRsingleObj} & \\
    \text{s.t.} & \vcx(t+1) = \mtbA \vcx(t) + \mtbB \vcu(t), & \quad \forall t \in \{0,1,\ldots\}
        \IEEEyessubnumber \label{subeq:LQRsingleDynamics} \\
                 &  \vcu(t) = \fnPhi(\vcx(t)), & \forall t \in \{0,1,\ldots\}
        \IEEEyessubnumber \label{subeq:LQRsingleController}
\end{IEEEeqnarray}
where $\fdPhi$ is the space of all functions $\fnPhi: \fdR^{F} \to \fdR^{G}$, see \cite{AndersonMoore89-LQR}. The objective function \eqref{subeq:LQRsingleObj} is the quadratic cost \eqref{eq:quadraticCostSingle}, the constraint \eqref{subeq:LQRsingleDynamics} imposes the linear dynamics of the system \eqref{eq:linearDynamicSingle} and the constraint \eqref{subeq:LQRsingleController} forces the solution to be a function $\fnPhi:\fdR^{F} \to \fdR^{G}$. The optimal controller obtained from solving \eqref{eq:LQRsingle} is formally known as a \emph{linear-quadratic regulator} (LQR) and is given by
% eq:LQRoptimalController
\begin{equation}
    \vcu^{\opt}(t) = \fnPhi^{\opt}\big( \vcx(t) \big) = \mtK^{\opt} \vcx(t),
\end{equation}
with $\mtK^{\opt} \in \fdR^{F \times G}$ being a linear operator that depends on the matrices that describe the problem, namely $\mtbA, \mtbB, \mtbQ, \mtbR$, and can be readily computed \cite[Sec.~2.4]{AndersonMoore89-LQR}. Notably, the LQR is a linear controller \cite[eq. (2.4-8)]{AndersonMoore89-LQR}.

A network system can be conveniently described by means of a graph $\stG = (\stV, \stE)$, where $\stV = \{v_{1},\ldots,v_{N}\}$ is the set of $N$ nodes and $\stE \subseteq \stV \times \stV$ is the set of edges. The node $v_{i}$ represents the $i^{\text{th}}$ component of the system, while the existence of the edge $(v_{i},v_{j}) \in \stE$ implies that nodes $v_{i}$ and $v_{j}$ are interconnected and capable of exchanging information. In a network system, each node is described by a state $\vcx_{i}(t) \in \fdR^{F}$ and is capable of autonomously taking an action $\vcu_{i}(t) \in \fdR^{G}$ at time $t$. The states and actions of all nodes are collected in two matrices $\mtX(t) \in \fdR^{N \times F}$ and $\mtU(t) \in \fdR^{N \times G}$, respectively, where each row corresponds to the state or action of each agent.

Similar to \eqref{eq:linearDynamicSingle}, consider a network system with linear dynamics modeled as
% eq:linearDynamic
\begin{equation} \label{eq:linearDynamic}
    \mtX(t+1) = \mtA \mtX(t) \mtbA + \mtB \mtU(t) \mtbB,
\end{equation}
where $\mtA \in \fdR^{N \times N}$ is called the \emph{network system matrix} and $\mtB \in \fdR^{N \times N}$ the \emph{network control matrix}.
The linear system in \eqref{eq:linearDynamic} is an extension of \eqref{eq:linearDynamicSingle} tailored to handle network data. In particular, it considers that each node $\lmv_{i}$ is described by an $F$-dimensional state $\vcx_{i}(t)$, collected in the rows of the matrix $\mtX(t)$. It also decouples the impact that the network topology has on the evolution of the system (through matrices $\mtA$ and $\mtB$) from the impact that the individual states have (through $\mtbA$ and $\mtbB$). To see this, note that matrices $\mtA \in \fdR^{N \times N}$ and $\mtB \in \fdR^{N \times N}$ act as linear combinations of state values across different nodes, and as such, these combinations are typically restricted to follow the interconnection of the components (although, technically, they need not be). It is thus noted that while the matrix $\mtA$ need not be the adjacency matrix of the graph, it is usually a function of it --for example, both matrices may share the same eigenvectors. The matrices $\mtbA \in \fdR^{F \times F}$ and $\mtbB \in \fdR^{G \times F}$ determine the evolution of the values of the state at each individual node and, while they can be arbitrary, they force all individual state nodes to follow the same evolution. Finally, it is noted that, while a more general linear description can be obtained by adopting a network state of dimension $NF$ and using \eqref{eq:linearDynamicSingle}, doing so obscures the effect of the topology of the network on the evolution of the system. Thus, \eqref{eq:linearDynamic} is adopted from now on for mathematical simplicity --and it is observed that all the results derived from here onward hold for \eqref{eq:linearDynamicSingle} as well.

To pose the linear-quadratic problem for a network system, the following quadratic cost as a counterpart of \eqref{eq:quadraticCostSingle} is adopted:
% eq:quadraticCost
\begin{equation} \label{eq:quadraticCost}
    \fnJ \Big( \{\mtX(t)\}, \{\mtU(t)\}\Big) =  \sum_{t=0}^{\infty} \Big( \| \mtX(t) \mtbQ^{1/2} \|_{F}^{2}+\| \mtU(t) \mtbR^{1/2} \|_{F}^{2} \Big),
\end{equation}
where $\mtbQ \in \fdR^{F \times F}$ and $\mtbR \in \fdR^{G \times G}$ are two given positive definite matrices, and where $\| \cdot \|_{F}$ denotes the Frobenius matrix norm. The linear-quadratic control problem for a network system can then be posed in the form of \eqref{eq:LQRsingle}, by replacing the cost \eqref{subeq:LQRsingleObj} with \eqref{eq:quadraticCost}, the linear dynamics \eqref{subeq:LQRsingleDynamics} with \eqref{eq:linearDynamic}, and the controller \eqref{subeq:LQRsingleController} with one such that $\fnPhi: \fdR^{N \times F} \to \fdR^{N \times G}$.

The controller solving the linear-quadratic problem for a network system is also linear. However, in order to compute the optimal control action, the system needs to access the state of arbitrary components of the system, beyond those directly connected. This constitutes a \emph{centralized} controller. In what follows, the focus is on finding a \emph{distributed} controller.

A distributed controller, which is denoted by $\fnPhi(\mtX(t);\stG)$ to emphasize its dependence on the topology of the network system $\stG$, should satisfy the properties that the control actions $\mtU(t)$ rely only on local information provided by other components that share a direct connection, and that they can be computed separately at each component. The use of a distributed controller overcomes some of the issues that arise when considering a centralized one. Namely, they are expected to scale better, since they do not require a single unit to compute the actions of all components in the system, and they are easy to implement since they do not demand an infrastructure capable of connecting all components to the single centralized unit.

The distributed linear-quadratic problem can be written as
% eq:distributedLQR
\begin{IEEEeqnarray}{CLR} \IEEEyesnumber \label{eq:distributedLQR}
    \!\!\min_{\fnPhi \in \fdPhi_{\stG}} & \!\fnJ \Big( \{\mtX(t)\} ,\{\mtU(t)\} \Big) &
    \IEEEyessubnumber \label{subeq:distributedLQRobj} \\
    \!\!\text{s.t.} & \!\mtX(t+1) \! = \! \mtA \mtX(t) \mtbA \! + \! \mtB \mtU(t) \mtbB, & \ \forall t \! \in \! \{0,1,\ldots\}
    \IEEEyessubnumber \label{subeq:distributedLQRdynamics} \\
    & \mtU(t) \! = \! \fnPhi\big( \mtX(t); \stG\big), & \forall t \! \in \! \{0,1,\ldots\}, \IEEEyessubnumber \label{subeq:distributedLQRcontroller}
\end{IEEEeqnarray}
where $\fdPhi_{\stG}$ is the space of all functions $\fnPhi(\cdot;\stG): \fdR^{N \times F} \to \fdR^{N \times G}$ that can be computed in a distributed manner (i.e. relying only on local information and computed separately at each component). It is noted that the constraint \eqref{subeq:distributedLQRcontroller} further restricts the feasible set, and as such, the optimal value $J_{\stG}^{\opt}$ of solving \eqref{eq:distributedLQR} is lower bounded by the optimal value $J^{\opt}$ incurred when using the optimal centralized controller, i.e. $J_{\stG}^{\opt} \geq J^{\opt}$.

Solving problem \eqref{eq:distributedLQR} requires solving an optimization problem over the space of functions $\fdPhi_{\stG}$. This is mathematically intractable in the general case, and requires specific approaches involving variational methods, dynamic programming or kernel-based functions \cite{Jahn07-NonlinearOptimization}. While there is a large class of distributed control problems that admit a convex formulation \cite{Rotkowitz06-DecentralizedConvex}, many of them lead to complex solutions that do not scale with the size of the network \cite{Fattahi19-TransformCentralized}.

Considering the inherent complexities of functional optimization, a popular approach is to adopt a specific model for the mapping $\fnPhi$, leading to a parametric family of controllers. Inspired by the linear nature of the optimal centralized solution and its mathematical tractability, a distributed linear parametrization was adopted in \cite{Fazelnia17-LowerBoundLQR, Fattahi19-LQR}. Many properties of this parametric family of controllers have been studied, including stability, robustness and (sub)optimality \cite{Fazelnia17-LowerBoundLQR, Fattahi19-LQR}.

However, it is known that the linear system \eqref{eq:linearDynamic} may have a nonlinear optimal controller if we force a distributed nature on its solution \cite{Witsenhausen68-Counterexample}. This suggests that it would be more convenient to work with nonlinear parametrizations, rather than linear ones. In particular, this work focuses on graph neural networks (GNNs) \cite{Gama20-GNNs}. These are nonlinear mappings that exhibit several desirable properties. Fundamentally, they are naturally computed in a distributed manner relying only on local information provided by directly connected components. This implies that any controller that is parametrized by means of a GNN respects the distributed nature of the system (as given by the graph $\stG$), naturally incorporating the distributed constraint \eqref{subeq:distributedLQRcontroller} into the chosen parametrization.

%%%%%%%%%%%%%%%%%%%%%%%%%%%%%%%%%%%%%%%%%%%%%%%%%%%%%%%%%%%%%%%%%%%%%%%%%%%%%%%%
%%%%                                                                        %%%%
%%%%                         GRAPH NEURAL NETWORKS                          %%%%
%%%%                                                                        %%%%
%%%%%%%%%%%%%%%%%%%%%%%%%%%%%%%%%%%%%%%%%%%%%%%%%%%%%%%%%%%%%%%%%%%%%%%%%%%%%%%%

\section{Graph Neural Networks} \label{sec:GNN}

%!TEX root = 00-distributedLQR.tex

%%%%%%%%%%%%%%%%%%%%%%%%%%%%%%%%%%%%%%%%%%%%%%%%%%%%%%%%%%%%%%%%%%%%%%%%%%%%%%%%
%%%%                                                                        %%%%
%%%%                         GRAPH NEURAL NETWORKS                          %%%%
%%%%                                                                        %%%%
%%%%%%%%%%%%%%%%%%%%%%%%%%%%%%%%%%%%%%%%%%%%%%%%%%%%%%%%%%%%%%%%%%%%%%%%%%%%%%%%
%%%% sec:GNN
%%%%%%%%%%%%

Finding the optimal distributed controller by solving problem \eqref{eq:distributedLQR} is intractable in its most general case. This is due to the constraint \eqref{subeq:distributedLQRcontroller} that the solution satisfies a distributed computation. In what follows, a parametric family of distributed controllers is adopted. More concretely, inspired by the fact that the optimal controller is usually nonlinear, GNN-based controllers are considered. The basics of graph signal processing are introduced in Section~\ref{subsec:GSP}, which allows for the definition of GNNs in Section~\ref{subsec:GCNN}. A discussion on how to cast the resulting finite-dimensional optimal control problem as an unsupervised learning problem follows in Section~\ref{subsec:selfsupervised}.

%%%%%%%%%%%%%%%%%%%%%%%%%%%%%%%%%%%%%%%%%%%%%%%%%%%%%%%%%%%%%%%%%%%%%%%%%%%%%%%%
%%%%                  SUBSECTION : Graph Signal Processing                  %%%%
%%%%%%%%%%%%%%%%%%%%%%%%%%%%%%%%%%%%%%%%%%%%%%%%%%%%%%%%%%%%%%%%%%%%%%%%%%%%%%%%
%%%% subsec:GSP
%%%%%%%%%%%%

\subsection{Graph signal processing} \label{subsec:GSP}

Graph signal processing (GSP) is a framework tailored to describe, analyze, and understand distributed problems \cite{Ortega18-GSP}. Given a graph $\stG = (\stV, \stE)$ that describes the structure of the data under study, a \emph{graph signal} $\fnx: \stV \to \fdR$ is defined as a mapping from the nodes of the graph to a real number. By imposing an arbitrary order on the nodes, this graph signal can be conveniently described as a vector $\vcx \in \fdR^{N}$ whose $i^{\text{th}}$ element corresponds to the signal value associated to node $v_{i}$, denoted as $[\vcx]_{i} = \fnx(v_{i}) = x_{i} \in \fdR$. Note that $[\cdot]_{i}$ ($[\cdot]_{ij}$) denotes the value of the $i^{\text{th}}$ ($(i,j)^{\text{th}}$) entry of a vector (matrix). To be able to use the concept of graph signals to describe the state $\mtX(t) \in \fdR^{N \times F}$ and the control action $\mtU(t) \in \fdR^{N \times G}$ in a network system, an extension to vector-valued mappings is needed. Define the \emph{vector-valued graph signal} as $\fnX: \stV \to \fdR^{F}$, where $\fnX(\lmv_{i}) = \vcx_{i} \in \fdR^{F}$. It can then be described by means of a matrix $\mtX \in \fdR^{N \times F}$, where each row corresponds to the signal value at each node
% eq:graphSignal
\begin{equation} \label{eq:graphSignal}
    \mtX =
    \begin{bmatrix}
        \fnX^{\Tr}(\lmv_{1}) \\
        \vdots \\
        \fnX^{\Tr}(\lmv_{N})
    \end{bmatrix}
    =
    \begin{bmatrix}
        \vcx_{1}^{\Tr} \\
        \vdots \\
        \vcx_{N}^{\Tr}
    \end{bmatrix}
    =
    \begin{bmatrix}
        \vcx^{1} & \cdots & \vcx^{F}
    \end{bmatrix}.
\end{equation}
In this equation, the vector-valued graph signal $\fnX$ is viewed as a collection of $F$ traditional scalar-valued graph signals $\{\fnx^{f}\}_{f=1}^{F}$, placed in the columns of the matrix. Observe that $\fnx$ stands for the graph signal as a function, $x_{i}$ stands for the scalar value adopted by node $\lmv_{i}$, $\vcx$ for the vector collecting all these values; likewise, $\fnX$ stands for the vector-valued graph signal and $\mtX$ for the matrix collecting all the states at all nodes. All these quantities are related to the graph signal that is used to describe the state of the system. The \emph{size of the vector-valued graph signal} is defined as
% eq:graphSignalNorm
\begin{equation} \label{eq:graphSignalNorm}
    \| \mtX \| = \| \mtX\|_{2,1} = \sum_{f=1}^{F} \| \vcx^{f} \|_{2}.
\end{equation}
The $L_{2,1}$ norm for matrices \eqref{eq:graphSignalNorm} is chosen as the size of the graph signal norm for both its robustness and its mathematical tractability. Note that if $F=1$, then $\|\mtX\| = \|\vcx\|_{2}$ as expected. Finally, note that, in what follows, the term ``graph signal'' is used indistinctly to refer to either vector-valued or scalar-valued ones. Note that the trajectories of system states $\{\mtX(t)\}$ and control actions $\{\mtU(t)\}$ can each be modeled as a sequence of graph signals, indexed by the time parameter $t$ ---also known as graph processes \cite{Gama19-GLLN}.

Describing a graph signal in terms of a matrix is convenient because it allows for easy mathematical manipulation. However, this causes the loss of the information related to the underlying graph support. To recover this information, the graph is described in terms of a \emph{support matrix} $\mtS \in \fdR^{N \times N}$ that respects the sparsity of the graph, i.e. $[\mtS]_{ij}=s_{ij}$ can be nonzero if and only if $i=j$ or $(\lmv_{j},\lmv_{i}) \in \stE$. Any matrix that satisfies this condition can be used as a support matrix and thus it is a design choice. Typical choices include the adjacency matrix, the Laplacian matrix, the Markov matrix, and their normalized counterparts \cite{Ortega18-GSP}. A linear mapping $\fnS:\fdR^{N \times F} \to \fdR^{N \times F}$ between graph signals that relates the input to the underlying graph support $\mtY = \fnS(\mtX) = \mtS \mtX$ can be defined, such that the $(i,f)^{\text{th}}$ entry $y_{i}^{f}$ of the matrix $\mtY$ (the value of the $f^{\text{th}}$ scalar graph signal at node $\lmv_{i}$) is computed as
% eq:graphShift
\begin{equation} \label{eq:graphShift}
    y_{i}^{f} = [\mtY]_{if} = [\mtS \mtX]_{if} = \sum_{j=1}^{N} [\mtS]_{ij} [\mtX]_{jf} = \sum_{j:\lmv_{j} \in \stN_{i} \cup \{v_{i}\}} s_{ij} x_{j}^{f},
\end{equation}
where $\stN_{i} = \{\lmv_{j} \in \stV: (\lmv_{j},\lmv_{i}) \in \stE\}$ is the set of nodes that share an edge with $\lmv_{i}$ and $[\mtX]_{jf} = [\vcx_{j}]_{f} = [\vcx^{f}]_{j} = x_{j}^{f}$, see \eqref{eq:graphSignal}. The last equality in \eqref{eq:graphShift} holds because of the sparsity pattern of the support matrix $\mtS$ and implies that the computation of the value of the output graph signal $\mtY$ at node $v_{i}$ only requires information relied by its neighbors. In this respect, one can then think of the pair $(\mtX,\mtS)$ as the complete graph data containing all the relevant information; however, only $\mtX$ is regarded as the \emph{actionable} variable (the signal), while the support $\mtS$ is considered given and fixed and is determined by the physical constraints of the network.

The support matrix $\mtS$ can be thought of as a linear mapping between graph signals that effectively relates the input to the underlying graph support. As such, the operation $\mtS \mtX$ becomes the basic building block of graph signal processing \cite{Ortega18-GSP}. A finite-impulse response (FIR) graph filter $\fnH: \fdR^{N \times F} \to \fdR^{N \times G}$ is a linear operation between two graph signals, defined as a polynomial on $\mtS$
% eq:graphFilter
\begin{equation} \label{eq:graphFilter}
    \mtY = \fnH(\mtX;\mtS, \stH) = \sum_{k=0}^{K} \mtS^{k} \mtX \mtH_{k},
\end{equation}
where $\stH = \{\mtH_{k} \in \fdR^{F \times G}, k = 0,\ldots,K\}$ is the set of filter taps $\mtH_{k}$ that characterize the filter response. The filter \eqref{eq:graphFilter} is linear in the input $\mtX$ and is capable of mapping between vector-valued graph signals of different dimensions (but defined on the same graph given by $\mtS$).

The graph filter is a naturally distributed operation, meaning that the output of filtering in \eqref{eq:graphFilter} can be computed separately by each node relying only on information provided by one-hop neighbors. To understand this, note that multiplications to the left of $\mtX$ carry out a linear combination of signal values across different nodes, and thus this matrix needs to respect the sparsity of the graph so that only values at neighboring nodes are combined. This is the case for $\mtS^{k} = \mtS^{k-1} \mtS$ which amounts to communicating $k$ times with the one-hop neighbors. Therefore, an FIR graph filter is a distributed linear operation since it requires only $K$ communication exchanges with one-hop neighbors. Multiplications to the right of $\mtX$, on the other hand, are linear combinations of signal values located at the same node, and can thus be arbitrary. In particular, \eqref{eq:graphFilter} imposes a weight-sharing scheme, where the signal values at all nodes are combined in the same way. Finally, note that \eqref{eq:graphFilter} is a compact notation for denoting the graph filtering operation but, in practice, the nodes do not need access to the full matrix $\mtS$. They only need access to the entries corresponding to their one-hop neighbors in order to compute the proper linear combination indicated in \eqref{eq:graphShift}. Thus, in practice, the nodes need not know the entire graph topology.

The FIR graph filter \eqref{eq:graphFilter} can be understood as a bank of $FG$ filters acting on scalar-valued graph signals, see \cite{Segarra17-Linear, Gama20-GNNs}. It can thus be characterized by its frequency response given by the collection of univariate polynomials
% eq:freqResponse
\begin{equation} \label{eq:freqResponse}
\Big\{h_{fg}(\lambda)= \sum_{k=0}^{K} [\mtH_{k}]_{fg} \lambda^{k} : \lambda \in [\lambda_{l},\lambda_{h}] \ ,\  f=1,\ldots,F \ , \ g=1,\ldots,G \Big\}.
\end{equation}
The values of $\lambda_{l}$ and $\lambda_{h}$ are determined by the specific problem under study, and are typically set to be the minimum and maximum eigenvalues of the given $\mtS$. However, they may be different if the problem requires the filters to be able to act on more than one graph, see Section~\ref{subsec:robustness}. In that case, it may be convenient to select the interval so that it contains all the eigenvalues of all the support matrices under consideration.

The characterization of the filter in terms of the frequency response \eqref{eq:freqResponse} allows for the definition of the \emph{size of the graph filter} as
% eq:graphFilterNorm
\begin{equation} \label{eq:graphFilterNorm}
    C_{\fnH} = \| \mtC_{\fnH} \|_{\infty} \quad \text{with} \quad \mtC_{\fnH} \in \fdR^{F \times G} : [\mtC_{\fnH}]_{fg} = \max_{\lambda \in [\lambda_{l},\lambda_{h}]} |h_{fg}(\lambda)|.
\end{equation}
In what follows, the focus is further set on a particular class of graph filters, known as \emph{Lipschitz} filters. The graph filter \eqref{eq:graphFilter} is said to be a Lipschitz filter if its frequency response \eqref{eq:freqResponse} satisfies that
% eq:LipschitzFilters
\begin{equation} \label{eq:LipschitzFilters}
    |h_{fg}(\lambda_{1}) - h_{fg}(\lambda_{2})| \leq \gamma_{fg} |\lambda_{1}-\lambda_{2}|, \quad \forall \lambda_{1},\lambda_{2} \in [\lambda_{l},\lambda_{h}],
\end{equation}
for some constant $\gamma_{fg} > 0$, for all $f\in\{1,\ldots,F\}$ and $G \in \{1,\ldots,G\}$. The Lipschitz constant $\Gamma_{\fnH}$ of the filter is computed as
% eq:graphFilterLipschitz
\begin{equation} \label{eq:graphFilterLipzchitz}
\scGamma_{\fnH} = \| \mtGamma_{\fnH} \|_{\infty} \quad \text{with} \quad \mtGamma_{\fnH} \in \fdR^{F \times G} : [\mtGamma_{\fnH}]_{fg} = \gamma_{fg},
\end{equation}
which is the infinity norm $\|\mtGamma_{\fnH}\|_{\infty}$ for a matrix $\mtGamma_{\fnH} \in \fdR^{F \times G}$ containing the corresponding Lipschitz constants of each individual filter (i.e. the maximum absolute row sum of the matrix).

%%%%%%%%%%%%%%%%%%%%%%%%%%%%%%%%%%%%%%%%%%%%%%%%%%%%%%%%%%%%%%%%%%%%%%%%%%%%%%%%
%%%%            SUBSECTION : Graph Convolutional Neural Networks            %%%%
%%%%%%%%%%%%%%%%%%%%%%%%%%%%%%%%%%%%%%%%%%%%%%%%%%%%%%%%%%%%%%%%%%%%%%%%%%%%%%%%
%%%% subsec:GCNN
%%%%%%%%%%%%

\subsection{Graph neural networks} \label{subsec:GCNN}

Graph filters are distributed, linear operations and, as such, are only capable of capturing linear relationships between input and output. However, the objective of this work is to learn nonlinear distributed controllers. Arguably, the most straightforward way of converting a graph filter into a nonlinear processing unit without affecting its distributed nature is to include a pointwise nonlinearity
% eq:graphPerceptron
\begin{equation} \label{eq:graphPerceptron}
    \mtY = \fnsigma \big( \fnH( \mtX; \mtS, \stH) \big),
\end{equation}
where $\fnsigma: \fdR \to \fdR$ is a nonlinearity applied pointwise to the entries of the graph signal obtained from applying the graph filter, i.e. $[\fnsigma(\mtX)]_{if} = \fnsigma([\mtX]_{if})$. The operation \eqref{eq:graphPerceptron} is known as a \emph{graph perceptron} \cite{Gama20-GNNs} and, since the nonlinearity $\fnsigma(\cdot)$ is applied pointwise to the entries of the graph signal, it retains the distributed nature of the graph filter.

The graph perceptron \eqref{eq:graphPerceptron} is a nonlinear processing unit, but it has a limited representation power. To overcome this, a \emph{graph convolutional neural network} $\fnPhi(\cdot; \mtS, \stH): \fdR^{N \times F} \to \fdR^{N \times G}$ is defined as a cascade of $L$ graph perceptron units
% eq:GCNN
\begin{IEEEeqnarray}{L} \IEEEyesnumber \label{eq:GCNN}
    \mtX_{\ell} = \fnsigma \big( \fnH_{\ell}( \mtX_{\ell-1}; \mtS, \stH_{\ell}) \big),
        \IEEEyessubnumber \label{subeq:GCNNlayer} \\
    \fnPhi(\mtX; \mtS, \stH) = \mtX_{L},
        \IEEEyessubnumber \label{subeq:GCNNoutput}
\end{IEEEeqnarray}
with $\stH = \cup_{\ell=1}^{L} \stH_{\ell}$. The input to the first layer is the graph signal $\mtX_{0} = \mtX$ and the output is collected at the last layer. The space of all possible representations obtained by using a GNN is characterized by the set of filter taps $\stH$, which contains the filter coefficients $\stH_{\ell} = \{\mtH_{\ell k} \in \fdR^{F_{\ell-1} \times F_{\ell}} \ , \ k = 0, 1, \ldots, K_{\ell}\}$ at each layer $\ell \in \{1,\ldots,L\}$. Note that $F_{0}=F$ and $F_{L} = G$. The nonlinear function $\fnsigma(\cdot)$, the number of layers $L$, the dimension of the graph signals at each layer $F_{\ell}$ and the number of filter taps at each layer $K_{\ell}$ are design choices and are typically referred to as \emph{hyperparameters} \cite{Bergstra11-Hyperparameter}.

%\blue{In analogy to \eqref{eq:graphFilterNorm}, the \emph{size of the GCNN} is defined as
%% eq:GCNNnorm
%\begin{equation} \label{eq:GCNNnorm}
%    C_{\fnPhi} = \prod_{\ell=1}^{L} C_{\fnH_{\ell}}
%\end{equation}
%%
%with $C_{\fnH_{\ell}}$ being the filter size of each of the layers $\fnH_{\ell}$.} Likewise, if the constitutive filters at all layers are Lipschitz filters, one can define the \emph{Lipschitz constant of the GCNN} as
%% eq:LipschitzGCNN
%\begin{equation} \label{eq:LipschitzGCNN}
%\Gamma_{\fnPhi} = \sup_{\ell \in \{1,\ldots,L\}} \Gamma_{\fnH_{\ell}}.
%\end{equation}
%%
%Recall that the filter coefficients are typically learned from data, and thus $C_{\fnPhi}$ and $\Gamma_{\fnPhi}$ depend on the training process. Hence, one can add either $C_{\fnPhi}$ or $\Gamma_{\fnPhi}$ (or both) as penalties to the loss function during training if a controller with a small size or Lipschitz constant is to be sought.

%%%%%%%%%%%%%%%%%%%%%%%%%%%%%%%%%%%%%%%%%%%%%%%%%%%%%%%%%%%%%%%%%%%%%%%%%%%%%%%%
%%%%                 SUBSECTION : Self-supervised Learning                  %%%%
%%%%%%%%%%%%%%%%%%%%%%%%%%%%%%%%%%%%%%%%%%%%%%%%%%%%%%%%%%%%%%%%%%%%%%%%%%%%%%%%
%%%% subsec:selfsupervised
%%%%%%%%%%%%

\subsection{Self-supervised learning} \label{subsec:selfsupervised}

The linear graph filter \eqref{eq:graphFilter} and the nonlinear GNN \eqref{eq:GCNN} have been introduced as naturally distributed parametrizations. By choosing to adopt one of these models for the to-be-learned controller, the focus is immediately set on a distributed mapping between the state and the action, turning the functional optimization problem \eqref{eq:distributedLQR} into the finite-dimensional optimization
% eq:distributedLQRparam
\begin{IEEEeqnarray}{CL} \IEEEyesnumber \label{eq:distributedLQRparam}
    \min_{\stH} & \fnJ \Big( \{\mtX(t)\} ,\{\mtU(t)\} \Big)
    \IEEEyessubnumber \label{subeq:distributedLQRparamObj} \\
    \text{s. t.} & \mtX(t+1) = \mtA \mtX(t) \mtbA + \mtB \mtU(t) \mtbB,
    \IEEEyessubnumber \label{subeq:distributedLQRparamDynamics} \\
    & \mtU(t) = \fnPhi\big( \mtX(t); \mtS, \stH\big) \IEEEyessubnumber \label{subeq:distributedLQRparamController}.
\end{IEEEeqnarray}
The constraint \eqref{subeq:distributedLQRparamController} replaces a generic distributed controller $\fnPhi(\mtX(t); \stG)$ in \eqref{subeq:distributedLQRcontroller} with a controller that admits a parametrization based on either a graph filter or a GNN. The resulting controller $\fnPhi(\mtX(t); \mtS, \stH^{\opt})$ with filter coefficients $\stH^{\opt}$ that solves \eqref{eq:distributedLQRparam} naturally satisfies the distributed constraint.

Problem \eqref{eq:distributedLQRparam} is nonconvex when adopting a GNN-based controller \eqref{subeq:distributedLQRparamController}. Thus, to approximately solve this problem, the empirical risk minimization (ERM) approach that is typical in learning theory \cite{Vapnik00-StatisticalLearning} is leveraged. To do this, a \emph{training set} $\stT = \{\mtX_{1,0}, \ldots, \mtX_{|\stT|,0}\}$ containing $|\stT|$ samples $\mtX_{p,0}$ drawn independently from some distribution $\fnp$ is considered to be the different random initializations of the system. Then, the ERM problem is given by
% eq:ERM
\begin{IEEEeqnarray}{CL} \IEEEyesnumber \label{eq:ERM}
    \min_{\stH} & \sum_{p=1}^{|\stT|} \fnJ \Big( \{\mtX_{p}(t)\} ,\{\mtU_{p}(t)\} \Big)
        \IEEEyessubnumber \label{subeq:ERMparamObj} \\
    \text{s. t.} & \mtX_{p}(t+1) = \mtA \mtX_{p}(t) \mtbA + \mtB \mtU_{p}(t) \mtbB,
        \IEEEyessubnumber \label{subeq:ERMdynamics} \\
    & \mtU_{p}(t) = \fnPhi\big( \mtX_{p}(t); \mtS, \stH\big) ,
        \IEEEyessubnumber \label{subeq:ERMcontroller} \\
    & \mtX_{p}(0) = \mtX_{p,0}.
        \IEEEyessubnumber \label{subeq:ERMinitial}
\end{IEEEeqnarray}
Problem \eqref{eq:ERM} can be solved by means of an algorithm based on stochastic gradient descent \cite{Kingma15-ADAM}, efficiently computing the gradient of $\fnJ(\cdot, \cdot)$ with respect to the parameter $\stH$ by means of the back-propagation algorithm \cite{Rumelhart86-BackProp}. To estimate the performance of the learned controllers --i.e. those obtained by solving \eqref{eq:ERM}-- a new set of initial states is generated, called the test set, and the average quadratic cost \eqref{eq:quadraticCost} is computed on the resulting trajectories. In essence, the optimization problem \eqref{eq:distributedLQRparam} is transformed into a self-supervised ERM problem \eqref{eq:ERM} that is solved through simulated data.

It is observed that, during the training phase, the optimization problem \eqref{eq:ERM} has to be solved in a centralized manner due to the weight-sharing scheme imposed by the FIR graph filters (recall that this weight-sharing scheme is necessary for scalability, keeping the number of learnable parameters independent of the size of the graph). However, this training phase can be carried out offline, prior to online execution. Once the GNN-based controllers are learned and the training phase is finished, they can be deployed in an entirely distributed manner for testing in the online phase. It is noted that there exist distributed optimization algorithms that leverage consensus to arrive to the optimal set of filter taps $\stH$ \cite{Nedich2020-DistributedOptimization}. These techniques, however, are outside the scope of the present work and will be left as future research directions.

%%%%%%%%%%%%%%%%%%%%%%%%%%%%%%%%%%%%%%%%%%%%%%%%%%%%%%%%%%%%%%%%%%%%%%%%%%%%%%%%
%%%%                                                                        %%%%
%%%%                     PROPERTIES OF GNN CONTROLLERS                      %%%%
%%%%                                                                        %%%%
%%%%%%%%%%%%%%%%%%%%%%%%%%%%%%%%%%%%%%%%%%%%%%%%%%%%%%%%%%%%%%%%%%%%%%%%%%%%%%%%

\section{Properties of GNN Controllers} \label{sec:GNNproperties}

%!TEX root = 00-distributedLQR.tex

%%%%%%%%%%%%%%%%%%%%%%%%%%%%%%%%%%%%%%%%%%%%%%%%%%%%%%%%%%%%%%%%%%%%%%%%%%%%%%%%
%%%%                                                                        %%%%
%%%%                     PROPERTIES OF GNN CONTROLLERS                      %%%%
%%%%                                                                        %%%%
%%%%%%%%%%%%%%%%%%%%%%%%%%%%%%%%%%%%%%%%%%%%%%%%%%%%%%%%%%%%%%%%%%%%%%%%%%%%%%%%
%%%% sec:GNNproperties
%%%%%%%%%%%%%%%%%%%%%%

GNNs have many suitable properties that make them appropriate choices for learning distributed controllers. As standalone processing units, they are naturally distributed architectures and have the properties of permutation equivariance and Lipschitz continuity to changes in the underlying graph support. As part of a linear dynamical system, GNN-based controllers can also be shown to stabilize the system. Furthermore, the deviation in the nominal trajectory due to unknown system matrices can be mitigated with properly learned filters. These properties, which are studied in this section, hold for any GNN controller of the form \eqref{eq:GCNN} that satisfy the corresponding hypotheses.

%%%%%%%%%%%%%%%%%%%%%%%%%%%%%%%%%%%%%%%%%%%%%%%%%%%%%%%%%%%%%%%%%%%%%%%%%%%%%%%%
%%%%                   SUBSECTION : Standalone Properties                   %%%%
%%%%%%%%%%%%%%%%%%%%%%%%%%%%%%%%%%%%%%%%%%%%%%%%%%%%%%%%%%%%%%%%%%%%%%%%%%%%%%%%
%%%% subsec:standalone
%%%%%%%%%%%%

\subsection{GNN properties} \label{subsec:standalone}

The main motivation for choosing GNNs as parametrizations for the controller is that they are naturally distributed architectures. GNNs are built by using graph filters and pointwise nonlinearities. Graph filters are distributed operations, as discussed after \eqref{eq:graphFilter}. The pointwise nonlinearity does not affect this, and thus GNNs are also distributed. It is noted that asynchronous implementations of graph filtering are possible \cite{Teke19-Asynchronous}. Additionally, GNNs are capable of learning nonlinear controllers, which is a key feature in the context of distributed control, as it is expected that optimal distributed controllers to be nonlinear \cite{Witsenhausen68-Counterexample}.

GNNs exhibit the property of permutation equivariance, \cite[Prop. 2]{Gama20-Stability}, which means that a reordering of the nodes does not affect the output, since it will be correspondingly reordered. This further implies that the GNNs are capable of leveraging any existing symmetries in the underlying graph topology to improve training. More specifically, learning how to process a given signal from the training set means that the GNN learns how to process the same signal anywhere in the graph with the same neighborhood topology. In a manner akin to the data augmentation that happens naturally by the choice of the convolution operation in regular convolutional neural networks (CNNs), permutation equivariance shows precisely one way in which the GNN exploits the data structure to improve training and generalization.

GNNs are also Lipschitz continuous to changes in the underlying graph \cite[Thm. 4]{Gama20-Stability}. This means that, if the underlying graph support is perturbed, the output of the GNN changes linearly with the size of perturbation. This implies that a GNN trained on one graph but tested on another one will still work well as long as both graphs are similar, see \cite{Ruiz20-Transferability}. It also implies that if the graph is not known exactly but has to be estimated, then the GNN can still be trained as long as the graph support estimate is good enough. Additionally, it indicates that GNNs are suitable for time-varying scenarios where the changes to the graph support are slow \cite{Gama21-ControlGNN}.

%Finally, it is noted that GNNs share many traits with regular CNNs. Namely, they are easy to train because they are close to linear, they generalize well due to the weight-sharing scheme that exploits data structure, and they have exhibited major success in several tasks including source localization, authorship attribution, recommender systems \cite{Gama20-GNNs}, and also in problems defined over physical networks such as prediction on sensor networks \cite{Owerko18-Sensor}, and robotics \cite{Gama21-ControlGNN}. All these characteristics render GNNs a suitable choice of parametrization for learning distributed controllers.

%%%%%%%%%%%%%%%%%%%%%%%%%%%%%%%%%%%%%%%%%%%%%%%%%%%%%%%%%%%%%%%%%%%%%%%%%%%%%%%%
%%%%                         SUBSECTION : Stability                         %%%%
%%%%%%%%%%%%%%%%%%%%%%%%%%%%%%%%%%%%%%%%%%%%%%%%%%%%%%%%%%%%%%%%%%%%%%%%%%%%%%%%
%%%% subsec:stability
%%%%%%%%%%%%

\subsection{Closed-Loop Stability} \label{subsec:stability}

GNNs have many suitable properties for learning distributed controllers. However, this does not necessarily guarantee that they are a good choice for a control system. In what follows, properties relating to GNN-based controllers within a linear dynamical system are studied.

A network system with the linear dynamics \eqref{eq:linearDynamic} is characterized by the set of matrices $\stD = \{\mtS, \mtA, \mtbA, \mtB, \mtbB\}$, where $\mtS \in \fdR^{N \times N}$ is the graph support matrix, $\mtA \in \fdR^{N \times N}$ and $\mtbA \in \fdR^{F \times F}$ are the system matrices, and $\mtB \in \fdR^{N \times N}$ and $\mtbB \in \fdR^{G \times F}$ are the control matrices. The trajectory of the system $\{\mtX(t)\}$ depends on these matrices. GNNs are capable of stabilizing the closed-loop dynamics of a distributed linear system $\stD$. More specifically, drawing from \cite{Lavaei20-StableRL},  the notion of input-state stability is defined as follows.

%%%%%%%%%%%%%%%%%%%%%%%%%%%%%%%%%%%%%%%%
%%%%           DEFINITION           %%%% def:stability
%%%%%%%%%%%%%%%%%%%%%%%%%%%%%%%%%%%%%%%%

\begin{definition}[Input-state stability] \label{def:stability}
    Consider a linear dynamical system as in \eqref{eq:linearDynamic} controlled by $\mtU(t) = \fnPhi(\mtX(t))+\mtE(t)$ where $\mtE(t)$ is a disturbance term or exploratory signal. The system is input-state stable if, for all sequences $\{\mtX(t)\}$ and $\{\mtE(t)\}$ such that $\sum_{t=0}^{\infty} \|\mtX(t)\| < \infty$ and $\sum_{t=0}^{\infty} \|\mtE(t)\| < \infty$, there exist constants $\beta_{0},\beta_{1} \geq 0$ such that
    % eq:stability
    \begin{equation} \label{eq:stability}
        \sum_{t=0}^{\infty} \| \mtX(t) \| \leq \beta_{0}+ \beta_{1} \sum_{t=0}^{\infty} \| \mtE(t)\|.
    \end{equation}
\end{definition}

%%%%%%%%%%%%%%%%%%%%%%%%%%%%%%%%%%%%%%%%
This definition of input-state stability is widely used \cite{Lavaei20-StableRL}. Given a trained GNN-based controller, a sufficient condition for the resulting system to be stable can be determined.

%%%%%%%%%%%%%%%%%%%%%%%%%%%%%%%%%%%%%%%%
%%%%            THEOREM             %%%% thm:stability
%%%%%%%%%%%%%%%%%%%%%%%%%%%%%%%%%%%%%%%%

\begin{theorem}[Sufficient condition for input-state stability]
        \label{thm:stability}
    Consider a distributed linear system $\stD$. Assume that the system is controlled with a GNN \eqref{eq:GCNN} consisting of $L$ layers of filters $\fnH_{\ell}(\cdot;\mtS,\stH)$ with $F_{\ell}$ features and $K_{\ell}$ taps each. Let the nonlinearity $\fnsigma(\cdot)$ be such that $|\fnsigma(x)| \leq |x|$. Then, the closed-loop system is input-state stable if it holds that
    % eq:stabilityCondition
    \begin{equation} \label{eq:stabilityCondition}
        \scxi (\stD, \stH) < 1,
    \end{equation}
    where
    % eq:stabilityConstant
    \begin{equation} \label{eq:stabilityConstant}
        \scxi (\stD, \stH) = \| \mtA \|_{2} \| \mtbA\|_{\infty} + C_{\fnPhi} \| \mtB\|_{2} \| \mtbB\|_{\infty}
    \end{equation}
    is the stability constant, with $C_{\fnPhi} = \prod_{\ell=1}^{L} C_{\fnH_{\ell}}$ for $C_{\fnH_{\ell}}$ the size of the $\ell^{\text{th}}$ filter, see \eqref{eq:graphFilterNorm}.
\end{theorem}
%%%%%%%%%%%%%%%% PROOF %%%%%%%%%%%%%%%%%
\begin{proof}
\ifundefined{arXiv}
    See~\ref{app:stability}.
\else
    See Appendix~\ref{app:stability}.
\fi
\end{proof}
%%%%
%%%%         End of THEOREM         %%%%
%%%%%%%%%%%%%%%%%%%%%%%%%%%%%%%%%%%%%%%%

Theorem~\ref{thm:stability} is a sufficient condition for the closed-loop system to be input-state stable. The learned filters affect the constant $C_{\fnPhi}$ such that the smaller the filters $C_{\fnH_{\ell}}$ the smaller $C_{\fnPhi}$ and thus $\scxi$. Therefore, a penalty on the size of the filters, see \eqref{eq:graphFilterNorm}, can be added to the objective function of \eqref{eq:ERM} to obtain GNNs with a controlled value of $C_{\fnPhi}$ and therefore with a smaller stability constant $\scxi$. The condition on the nonlinearity is mild and is satisfied by the most popular nonlinearities ($\ReLU$, $\tanh$, $\sigmoid$, etc.). It is observed that the sufficient condition requires $\|\mtA\|_{2}\|\mtbA\|_{\infty} < 1$, which implies that the system is open-loop stable. In many physical systems such as power networks, it is possible to design stabilizing controllers. This implies that once the system has been stabilized a GNN-based controller can then be learned to minimize the quadratic cost.

%%%%%%%%%%%%%%%%%%%%%%%%%%%%%%%%%%%%%%%%%%%%%%%%%%%%%%%%%%%%%%%%%%%%%%%%%%%%%%%%
%%%%                        SUBSECTION : Robustness                         %%%%
%%%%%%%%%%%%%%%%%%%%%%%%%%%%%%%%%%%%%%%%%%%%%%%%%%%%%%%%%%%%%%%%%%%%%%%%%%%%%%%%
%%%% subsec:robust
%%%%%%%%%%%%

\subsection{Trajectory deviation} \label{subsec:robustness}

It often happens that one does not have direct access to the matrices $\stD$ that characterize the distributed linear system and thus they should be estimated. Alternatively, sometimes the system description may change slightly from the training to the testing phase. Therefore, it is essential to study the impact of the inaccurate knowledge of these matrices on the trajectory.

Consider a network system on a graph $\stG$ with the linear dynamics \eqref{eq:linearDynamic} and described by the set of matrices $\stD$. Assume that these matrices are unknown and, instead, access to estimates of these matrices is provided. These estimates are denoted by $\sthD = \{\mthS, \mthA, \mthbA, \mthB, \mthbB\}$ where $\mthS \in \fdR^{N \times N}$ is the estimate of the support matrix (i.e. the exact graph support is unknown), $\mthA \in \fdR^{N \times N}$ and $\mthbA \in \fdR^{F \times F}$ are the estimates of the system matrices, and $\mtbB \in \fdR^{N \times N}$ and $\mthbB \in \fdR^{G \times F}$ are the estimates of the control matrices. It is evident that the trajectory $\{\mthX(t)\}$ on the linear dynamical network $\sthD$ could be noticeably different from $\{\mtX(t)\}$, the one obtained from the system described by $\stD$.

The goal is to characterize how the difference in the systems $\stD$ and $\sthD$ impacts their respective trajectories $\{\mtX(t)\}$ and $\{\mthX(t)\}$. Towards this end, a notion of distance between the system matrices is first defined.
%
%%%%%%%%%%%%%%%%%%%%%%%%%%%%%%%%%%%%%%%%
%%%%           DEFINITION           %%%% def:systDistance
%%%%%%%%%%%%%%%%%%%%%%%%%%%%%%%%%%%%%%%%
%
\begin{definition}[Distance between systems] \label{def:systDistance}
    Given the system matrices $\stD$ and $\sthD$, the \emph{distance between system descriptions} is defined as
    % eq:systDistance
    \begin{equation} \label{eq:systDistance}
        \fnd(\stD,\sthD) = \sceps,
    \end{equation}
    where $\sceps > 0$ is the smallest number such that
    % eq:systDistanceDef
    \begin{equation} \label{eq:systDistanceDef}
    \begin{gathered}
        \| \mtS - \mthS\|_{2} \leq \sceps \ , \ \|\mtA - \mthA\|_{2} \leq \sceps \ , \ \| \mtbA - \mthbA\|_{\infty} \leq \sceps, \\
        \| \mtB - \mthB\|_{2} \leq \sceps \ , \ \| \mtbB - \mthbB \|_{\infty} \leq \sceps.
    \end{gathered}
    \end{equation}
\end{definition}
%
%%%%%%%%%%%%%%%%%%%%%%%%%%%%%%%%%%%%%%%%
%

In other words, Definition \ref{def:systDistance} determines the distance between two system descriptions as the maximum norm difference in the constitutive matrix norms, with matrices on the graph domain being determined by the spectral norm $\|\cdot\|_{2}$, and matrices on the feature domain being determined by the infinity norm $\|\cdot\|_{\infty}$.

First, a result on how the input-state stability of the closed-loop system is affected by the distance between $\stD$ and $\sthD$ is obtained.
%
%%%%%%%%%%%%%%%%%%%%%%%%%%%%%%%%%%%%%%%%
%%%%          PROPOSITION           %%%% prop:stabilityChange
%%%%%%%%%%%%%%%%%%%%%%%%%%%%%%%%%%%%%%%%
%
\begin{proposition}[Change in input-state stability]
    \label{prop:stabilityChange}
    Consider two systems described by the sets of matrices $\stD$ and $\sthD$. Let these systems be controlled by a GNN \eqref{eq:GCNN} consisting of $L$ layers of filters $\fnH_{\ell}(\cdot;\cdot,\stH)$ with $F_{\ell}$ features and $K_{\ell}$ filter taps each. Let the nonlinearity $\sigma(\cdot)$ be such that $|\fnsigma(a) - \fnsigma(b)| \leq |a-b|$ and $\fnsigma(0)=0$. Then, it holds that
    % eq:stabilityChange
    \begin{equation} \label{eq:stabilityChange}
        \big| \scxi - \schxi \big| \leq \schC_{\scxi} \ \fnd(\stD, \sthD),
    \end{equation}
    where $\scxi = \scxi(\stD, \stH)$ and $\schxi = \scxi (\sthD, \stH)$ are the stability constants of the system $\stD$ and $\sthD$, respectively, and where
    % eq:Cxi
    \begin{equation} \label{eq:Cxi}
        \schC_{\scxi} =  \| \mtA \|_{2} + \| \mthbA \|_{\infty} + \scC_{\fnPhi} \big(  \| \mtB \|_{2}+\| \mthbB \|_{\infty} \big),
    \end{equation}
    with $C_{\fnPhi} = \prod_{\ell=1}^{L} C_{\fnH_{\ell}}$ for $C_{\fnH_{\ell}}$ the size of the $\ell^{\text{th}}$ filter, see \eqref{eq:graphFilterNorm}.
\end{proposition}
%%%%%%%%%%%%%%%% PROOF %%%%%%%%%%%%%%%%%
\begin{proof}
\ifundefined{arXiv}
    See~\ref{app:stability}.
\else
    See Appendix~\ref{app:stability}.
\fi
\end{proof}
%%%%
%%%%         End of THEOREM         %%%%
%%%%%%%%%%%%%%%%%%%%%%%%%%%%%%%%%%%%%%%%
%
Proposition~\ref{prop:stabilityChange} states that the difference in the stability constants between the system $\stD$ and its estimate $\sthD$ depends on the distance $\fnd(\stD,\sthD)$ between them, on the system matrices of both $\stD$ and $\sthD$, and on the learned filters through $\scC_{\fnPhi}$. If the matrix description of $\stD$ is inaccessible, then Def.~\ref{def:systDistance} can be leveraged to replace $\|\mtA\|_{2}$ and $\|\mtB\|_{2}$ in \eqref{eq:Cxi} by the upper bounds $\|\mtA\|_{2} \leq \|\mthA\|_{2} + \fnd(\stD,\sthD)$ and $\|\mtB\|_{2} \leq \|\mthB\|_{2} + \fnd(\stD, \sthD)$, respectively. The same holds if $\sthD$ is not known but $\stD$ is. It is also noted that, for the case when $F=G=1$, it follows from the proof that $\schC_{\scxi} = 1+\scC_{\fnPhi}$ and the bound is proportional to the distance $\fnd(\stD,\sthD)$; see \ifundefined{arXiv}~\ref{app:stability}.\else Appendix~\ref{app:stability}. \fi

Next, the goal is to characterize the deviation in the trajectories, namely $\|\mtX(t) - \mthX(t)\|$, as a function of how different the systems $\stD$ and $\sthD$ are. In this context, a controller $\fnPhi$ is acceptable if the resulting closed-loop trajectories of two different systems are similar as long as the systems themselves are similar. This is the case for GNN-based distributed controllers as shown next.
%
%%%%%%%%%%%%%%%%%%%%%%%%%%%%%%%%%%%%%%%%
%%%%            THEOREM             %%%% thm:robust
%%%%%%%%%%%%%%%%%%%%%%%%%%%%%%%%%%%%%%%%
%
\begin{theorem}[Bound on trajectory deviation]
    \label{thm:robust}
    Consider two systems described by the sets of matrices $\stD$ and $\sthD$. Let these systems be controlled by a GNN \eqref{eq:GCNN} consisting of $L$ layers of filters $\fnH_{\ell}(\cdot;\cdot,\stH)$ with $F_{\ell}$ features and $K_{\ell}$ filter taps each. Let the nonlinearity $\fnsigma(\cdot)$ be such that $|\fnsigma(a) - \fnsigma(b)| \leq |a-b|$ and $\fnsigma(0)=0$. Then, it holds that
    % eq:robust
    \begin{equation} \label{eq:robust}
        \big\| \mtX(t) - \mthX(t) \big\| \leq \schC_{\fnPhi} \schC_{t} \| \mtX(0)\|\ \fnd(\stD,\sthD),
    % \big\| \mtX(t) - \mthX(t) \big\| \leq \schC_{\fnPhi} \schC_{t} \| \mtX(0)\|\ \fnd(\stD,\sthD)
    \end{equation}
    with $\schC_{\fnPhi} =  \schC_{\scxi} + C_{\fnPhi} \scGamma_{\fnPhi} \|\mthB\|_{2} \|\mthbB\|_{\infty} (1+8\sqrt{N})$ for $\schC_{\scxi}$ as in \eqref{eq:Cxi}, $C_{\fnPhi} = \prod_{\ell=1}^{L}C_{\fnH_{\ell}}$ and $\scGamma_{\fnPhi} = \sum_{\ell=1}^{L} (\scGamma_{\fnH_{\ell}}/C_{\fnH_{\ell}})$ for $C_{\fnH_{\ell}}$ and $\scGamma_{\fnH_{\ell}}$ the size and Lipschitz constant of the $\ell^{\text{th}}$ filter, respectively, see \eqref{eq:graphFilterNorm} and \eqref{eq:graphFilterLipzchitz}; and with $\schC_{t}$ such that $\schC_{0} = 0$ and
    % eq:Ct
    \begin{equation} \label{eq:Ct}
        \schC_{t} = t \max\{\scxi, \schxi\}^{t-1}
    \end{equation}
    for $t \geq 1$, where $\scxi$ and $\schxi$ are the stability constants of the systems $\stD$ and $\sthD$, respectively, as in \eqref{eq:stabilityConstant}.
\end{theorem}
%%%%%%%%%%%%%%%% PROOF %%%%%%%%%%%%%%%%%
\begin{proof}
\ifundefined{arXiv}
    See~\ref{app:robust}.
\else
    See Appendix~\ref{app:robust}.
\fi
\end{proof}
%%%%
%%%%         End of THEOREM         %%%%
%%%%%%%%%%%%%%%%%%%%%%%%%%%%%%%%%%%%%%%%
%
Theorem \ref{thm:robust} states that, for a linear dynamical network system under a GNN-based distributed controller, the change in trajectory between the system $\stD$ and its estimated description $\sthD$ depends on the value of $\schC_{\fnPhi}$ that is independent of time, on the value of $\schC_{t}$ that is time-varying, and on their distance $\fnd(\stD,\sthD)$. The value of $\schC_{\fnPhi}$ is affected by the given system (through matrices in the estimated system $\sthD$ and the number of nodes $N$) and the resulting trained filters in the GNN (through $C_{\fnPhi}$ and $\scGamma_{\fnPhi}$). The value of $\schC_{t}$ is determined by the stability constants $\scxi$ and $\schxi$, and becomes larger as time passes if $\max\{\scxi, \schxi\} \geq 1$, but otherwise decreases for large $t$. Recall that $\scxi$ can be estimated from $\schxi$ by leveraging Proposition~\ref{prop:stabilityChange}. It is noted that the constants $\schC_{\fnPhi}$ and $\schC_{t}$ can be affected by judicious training. For example, by penalizing the size of the filters $C_{\fnH_{\ell}}$ and their Lipschitz constant $\scGamma_{\fnH_{\ell}}$ during training, the learned GNN-based controller can be forced to be more stable, see Section~\ref{sec:sims} for concrete examples.

For the particular case when the closed-loop system and its estimate are guaranteed to be input-state stable, the following corollary can be stated.
%
%%%%%%%%%%%%%%%%%%%%%%%%%%%%%%%%%%%%%%%%
%%%%           COROLLARY            %%%% cor:robustStable
%%%%%%%%%%%%%%%%%%%%%%%%%%%%%%%%%%%%%%%%
%
\begin{corollary}[Bound on trajectory deviation for stable systems]
    \label{cor:robustStable}
    Consider a system $\stD$ and its estimate $\sthD$ such that both satisfy Theorem~\ref{thm:stability}. Then, it holds that
    % eq:robustStable
    \begin{equation} \label{eq:robustStable}
        \big\| \mtX(t) - \mthX(t) \big\| \leq\schC \| \mtX(0)\|\ \fnd(\stD, \sthD),
    \end{equation}
    where $\schC = -e^{-1}\schC_{\fnPhi}/(\max\{\scxi,\schxi\} \times \log(\max\{\scxi,\schxi\}))$ and $\schC_{\fnPhi}$ is given in Theorem~\ref{thm:robust}. Furthermore, it holds that
    % eq:robustStableLimit
    \begin{equation} \label{eq:robustStableLimit}
    \lim_{t \to \infty }\big\| \mtX(t) - \mthX(t) \big\| = 0.
    \end{equation}
\end{corollary}
%%%%%%%%%%%%%%%% PROOF %%%%%%%%%%%%%%%%%
\begin{proof}
\ifundefined{arXiv}
    See~\ref{app:robust}.
\else
    See Appendix~\ref{app:robust}.
\fi
\end{proof}
%%%%
%%%%         End of THEOREM         %%%%
%%%%%%%%%%%%%%%%%%%%%%%%%%%%%%%%%%%%%%%%
%
It follows from Corollary~\ref{cor:robustStable} that if a system and its estimate are guaranteed to be input-state stable, then the trajectory deviation between both systems is bounded by a constant that is proportional to the distance between them and is independent of time $t$. Furthermore, this deviation is guaranteed to go to zero as $t$ increases.

%%%%%%%%%%%%%%%%%%%%%%%%%%%%%%%%%%%%%%%%%%%%%%%%%%%%%%%%%%%%%%%%%%%%%%%%%%%%%%%%
%%%%                                                                        %%%%
%%%%                         NUMERICAL EXPERIMENTS                          %%%%
%%%%                                                                        %%%%
%%%%%%%%%%%%%%%%%%%%%%%%%%%%%%%%%%%%%%%%%%%%%%%%%%%%%%%%%%%%%%%%%%%%%%%%%%%%%%%%

\section{Numerical Experiments} \label{sec:sims}

%!TEX root = 00-distributedLQR.tex

%%%%%%%%%%%%%%%%%%%%%%%%%%%%%%%%%%%%%%%%%%%%%%%%%%%%%%%%%%%%%%%%%%%%%%%%%%%%%%%%
%%%%                                                                        %%%%
%%%%                         NUMERICAL EXPERIMENTS                          %%%%
%%%%                                                                        %%%%
%%%%%%%%%%%%%%%%%%%%%%%%%%%%%%%%%%%%%%%%%%%%%%%%%%%%%%%%%%%%%%%%%%%%%%%%%%%%%%%%
%%%% sec:sims
%%%%%%%%%%%%%

In this section, numerical simulations illustrate the performance of GNN-based controllers in a distributed linear-quadratic problem. In particular, problem \eqref{eq:distributedLQR} is solved with $F=G=1$ so that $\mtbA$ and $\mtbB$ become scalars that are subsumed into matrices $\mtA$ and $\mtB$, respectively.

%%%%%%%%%%%%%%%%%%%%%%%%%%%%%%%%%%%%%%%%%%%%%%%%%%%%%%%%%%%%%%%%%%%%%%%%%%%%%%%%
%%%%                                SETTING                                 %%%%
%%%%%%%%%%%%%%%%%%%%%%%%%%%%%%%%%%%%%%%%%%%%%%%%%%%%%%%%%%%%%%%%%%%%%%%%%%%%%%%%

\medskip\noindent \textbf{Problem setup.} The system has $N$ nodes placed uniformly at random on the $[0,1] \times [0,1]$ plane. Edges are drawn between the $5$-nearest neighbors of each node. The support matrix $\mtS$ is considered to be the adjacency matrix, normalized by the largest eigenvalue so that $\|\mtS\|_{2}=1$. The network system matrix $\mtA$ and network control matrix $\mtB$ share the same eigenvectors with $\mtS$ and the diagonal elements are chosen randomly with a standard Gaussian distribution and are normalized so that $\|\mtA\|_{2} = 0.995$ and $\|\mtB\|_{2} = 1$. The cost matrices are set to $\mtQ = \mtR = \mtI$. Trajectories of length $T=50$ are simulated. Unless otherwise specified, the networks have $N=50$ nodes.

%%%%%%%%%%%%%%%%%%%%%%%%%%%%%%%%%%%%%%%%%%%%%%%%%%%%%%%%%%%%%%%%%%%%%%%%%%%%%%%%
%%%%                              CONTROLLERS                               %%%%
%%%%%%%%%%%%%%%%%%%%%%%%%%%%%%%%%%%%%%%%%%%%%%%%%%%%%%%%%%%%%%%%%%%%%%%%%%%%%%%%

\medskip\noindent \textbf{Controllers.} Five controllers are studied. (i: Optim) The optimal centralized controller is used as a baseline \cite[eq. (2.4-8)]{AndersonMoore89-LQR}. (ii: MLP) A centralized controller can be learned by using a multi-layer perceptron (MLP) with $NF_{\text{MLP}}$ units in the hidden layer, and $N$ units in the readout layer \cite{Capella03-DistributedNN}. (iii: D-MLP) As a comparative method, the learnable, distributed controller proposed in \cite{Huang05-LargeScaleDecentralized} is used; recall that this method learns a separate MLP for each node, particularly a hidden layer with $F_{\text{D-MLP}}$ units and a single output unit to estimate the control action of the node. (iv: GNN) A two-layer GNN \eqref{eq:GCNN} with $F_{1}$ features and $K_{1}$-order polynomials for the first layer and $F_{2}=1$ and $K_{2}=0$ for the second layer. (v: GF) A $K_{1}$-order polynomial graph filter with $F_{1}$ features \eqref{eq:graphFilter}, followed by a readout layer which is another graph filter with $F_{2}=1$ output features and $K_{2}=0$ filter taps, see \cite{Fattahi19-LQR}. For the nonlinear methods (ii)-(iv), the function $\tanh$ is applied pointwise between the first and the second layers.

%%%%%%%%%%%%%%%%%%%%%%%%%%%%%%%%%%%%%%%%%%%%%%%%%%%%%%%%%%%%%%%%%%%%%%%%%%%%%%%%
%%%%                        TRAINING AND EVALUATION                         %%%%
%%%%%%%%%%%%%%%%%%%%%%%%%%%%%%%%%%%%%%%%%%%%%%%%%%%%%%%%%%%%%%%%%%%%%%%%%%%%%%%%

\medskip\noindent \textbf{Training and evaluation.} The controllers (ii)-(v) are trained by solving the equivalent ERM problem \eqref{eq:ERM} over a generated training set consisting of $|\stT|=500$ initial states. The ADAM algorithm \cite{Kingma15-ADAM} with the learning rate $\mu$ and forgetting factors $0.9$ and $0.999$ is used to update the gradients over batches of $20$ trajectories. A validation stage leveraging a set of $50$ new, independent initial states is computed every $5$ training updates. After $30$ epochs of training, the parameters that exhibited the best performance during the validation stage are retained. The controllers are evaluated by computing the quadratic cost over trajectories obtained from a set of $50$ new, independent initial states. For ease of exposition, the resulting cost is normalized by the lower bound for the distributed linear-quadratic problem obtained in \cite{Fazelnia17-LowerBoundLQR}. The training and evaluation process is repeated for $100$ different realizations of the system matrices $\stD$. Median and standard deviation values of the normalized cost are reported.

%%%%%%%%%%%%%%%%%%%%%%%%%%%%%%%%%%%%%%%%
%%%%             TABLE              %%%%  tab:hparam
%%%%%%%%%%%%%%%%%%%%%%%%%%%%%%%%%%%%%%%%
%%
\begin{table*}[t]
    \small
    \centering
    \caption{Normalized cost of the distributed controllers. (a) Distributed controller (iv: GNN) for $\mu = 0.01$. (b) Distributed controller (v: GF) for $\mu = 0.005$. Lower bound: $65 (\pm 2)$.}
    \subfloat[GNN (iv: GNN)]{
    \begin{tabular}{c|ccc}
        $F$/$K$ & $2$         &       $3$         &       $4$         \\ \hline
        $16$ & $1.1396 (\pm 0.0379)$ & $1.1311 (\pm 0.0338)$ & $\mathbf{1.1052 (\pm 0.0295)}$ \\
        $32$ & $1.1440 (\pm 0.0348)$ & $1.1286 (\pm 0.0275)$ & $1.1354 (\pm 0.0255)$ \\
        $64$ & $1.1409 (\pm 0.0356)$ & $1.1300 (\pm 0.0272)$ & $1.1196 (\pm 0.0323)$ \\
    \end{tabular}
    \label{subtab:hParam:GCNN}
    }
    \hfil
    \subfloat[Graph Filter (v: GF)]{
    \begin{tabular}{c|ccc}
        $F/K$ &      $2$         &       $3$         &       $4$         \\ \hline
        $16$ & $1.1716 (\pm 0.0319)$ & $1.1449 (\pm 0.0331)$ & $1.1295 (\pm 0.0289)$ \\
        $32$ & $1.1609 (\pm 0.0291)$ & $1.1385 (\pm 0.0358)$ & $1.1233 (\pm 0.0285)$ \\
        $64$ & $1.1466 (\pm 0.0361)$ & $1.1248 (\pm 0.0313)$ & $\mathbf{1.1175 (\pm 0.0251)}$ \\
    \end{tabular}
    \label{subtab:hParam:GF}
    }
    \label{tab:hParam}
\end{table*}
%%
%%%%          End of TABLE          %%%%
%%%%%%%%%%%%%%%%%%%%%%%%%%%%%%%%%%%%%%%%

%%%%%%%%%%%%%%%%%%%%%%%%%%%%%%%%%%%%%%%%%%%%%%%%%%%%%%%%%%%%%%%%%%%%%%%%%%%%%%%%
%%%%                        HYPERPARAMETER SELECTION                        %%%%
%%%%%%%%%%%%%%%%%%%%%%%%%%%%%%%%%%%%%%%%%%%%%%%%%%%%%%%%%%%%%%%%%%%%%%%%%%%%%%%%

\medskip\noindent \textbf{Experiment 1: Design hyperparameters.} The first experiment studies the performance of the controllers (iv: GNN) and (v: GF) as a function of the number of features at the output of the first layer $F \in \{16,32,64\}$, and the order of the polynomial $K \in \{2, 3, 4\}$. The learning rate is chosen from the set $\mu \in \{0.005, 0.01, 0.05\}$ and the one yielding the best performance for each architecture is shown in Table~\ref{tab:hParam}. In general, the performance does not vary significantly as a function of the hyperparameters, with a difference of $3.8$ percentage points for (iv: GNN) and $5.4$ for (v: GF). From now on, the hyperparameter values are set to $F_{1} = 16$, $K_{1} = 4$ and $\mu = 0.01$ for (iv: GNN), and $F_{1} = 64$, $K_{1} = 4$ and $\mu = 0.005$ for (v: GF). The fact that $K_{1} = 4$ exhibits the best performance for both controllers evidences the importance of repeated communication with one-hop neighbors for collecting information farther away.

%%%%%%%%%%%%%%%%%%%%%%%%%%%%%%%%%%%%%%%%%%%%%%%%%%%%%%%%%%%%%%%%%%%%%%%%%%%%%%%%
%%%%                               COMPARISON                               %%%%
%%%%%%%%%%%%%%%%%%%%%%%%%%%%%%%%%%%%%%%%%%%%%%%%%%%%%%%%%%%%%%%%%%%%%%%%%%%%%%%%

\medskip\noindent \textbf{Experiment 2: Comparison.} For the second experiment, the performance of the controllers (iv: GNN) and (v: GF) is compared to that of the centralized baselines (i: Optim) and (ii: MLP), and that of the distributed method (iii: D-MLP). The hyperparameters of (ii: MLP) and (iii: D-MLP) are set to $(F_{\text{MLP}},\mu) = (16, 0.005)$ and $(F_{\text{D-MLP}}, \mu) = (16, 0.01)$, respectively, chosen for yielding the best performance from the set $\{16, 32, 64\}$ for the features and $\{0.005, 0.01, 0.05\}$ for the learning rate. The controller (ii: MLP) learns $80,000$ parameters and the controller (iii: D-MLP) learns $3,200$, while (iv: GNN) learns $80$ parameters and (v: GF) learns $320$. The centralized controllers (i: Optim) and (ii: MLP) exhibit a normalized cost of $0.9961 (\pm 0.0001)$ and $0.9969 (\pm 0.0003)$, respectively. This shows that these two controllers are better than any possible distributed one. The distributed method (iii: D-MLP) yields a cost of $1.0999 (\pm 0.0167)$, $0.5$ percentage points better than (iv: GNN) which shows a cost of $1.1052(\pm0.0295)$ and $1.7$ percentage points better than (v: GF) which shows a cost of $1.1175 (\pm 0.0251)$. Overall, as expected, the centralized controllers perform better than the distributed ones. The performance of the controller (iii: D-MLP) is slightly better than (iv: D-MLP), possibly due to the fact that (iii: D-MLP) exhibits a larger representation space that can be successfully navigated given the rich training setting available in this simulation. It is observed in experiments~3~and~4, however, that this controller is not robust to changes in the underlying topology nor scales well, precisely due to the large number of parameters. Finally, it is observed that the nonlinear distributed controllers (iii) and (iv) outperform the linear one (v: GF).

%%%%%%%%%%%%%%%%%%%%%%%%%%%%%%%%%%%%%%%%%%%%%%%%%%%%%%%%%%%%%%%%%%%%%%%%%%%%%%%%
%%%%                            OPEN-LOOP SYSTEM                            %%%%
%%%%%%%%%%%%%%%%%%%%%%%%%%%%%%%%%%%%%%%%%%%%%%%%%%%%%%%%%%%%%%%%%%%%%%%%%%%%%%%%

%%%%%%%%%%%%%%%%%%%%%%%%%%%%%%%%%%%%%%%%
%%%%             FIGURE             %%%%  fig:normOpenLoop
%%%%%%%%%%%%%%%%%%%%%%%%%%%%%%%%%%%%%%%%
%% {subfig:normStable, subfig:normUnstable}
\begin{figure*}[!t]
    \centering
    \subfloat[Stable open-loop system]{
        \includegraphics[width=0.475\textwidth]{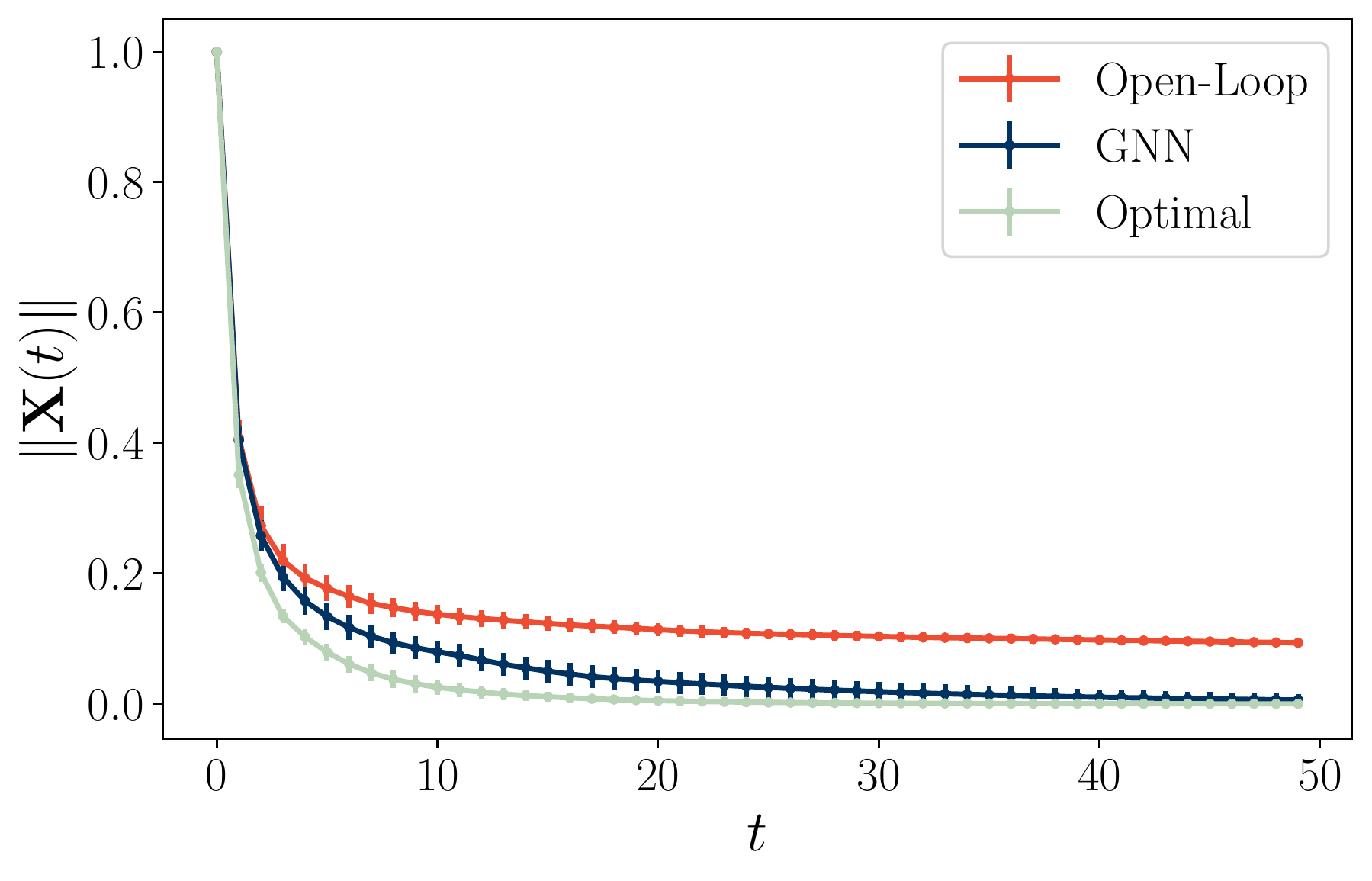}%
        \label{subfig:normStable}
    }
    \hfil
    \subfloat[Unstable open-loop system]{
        \includegraphics[width=0.475\textwidth]{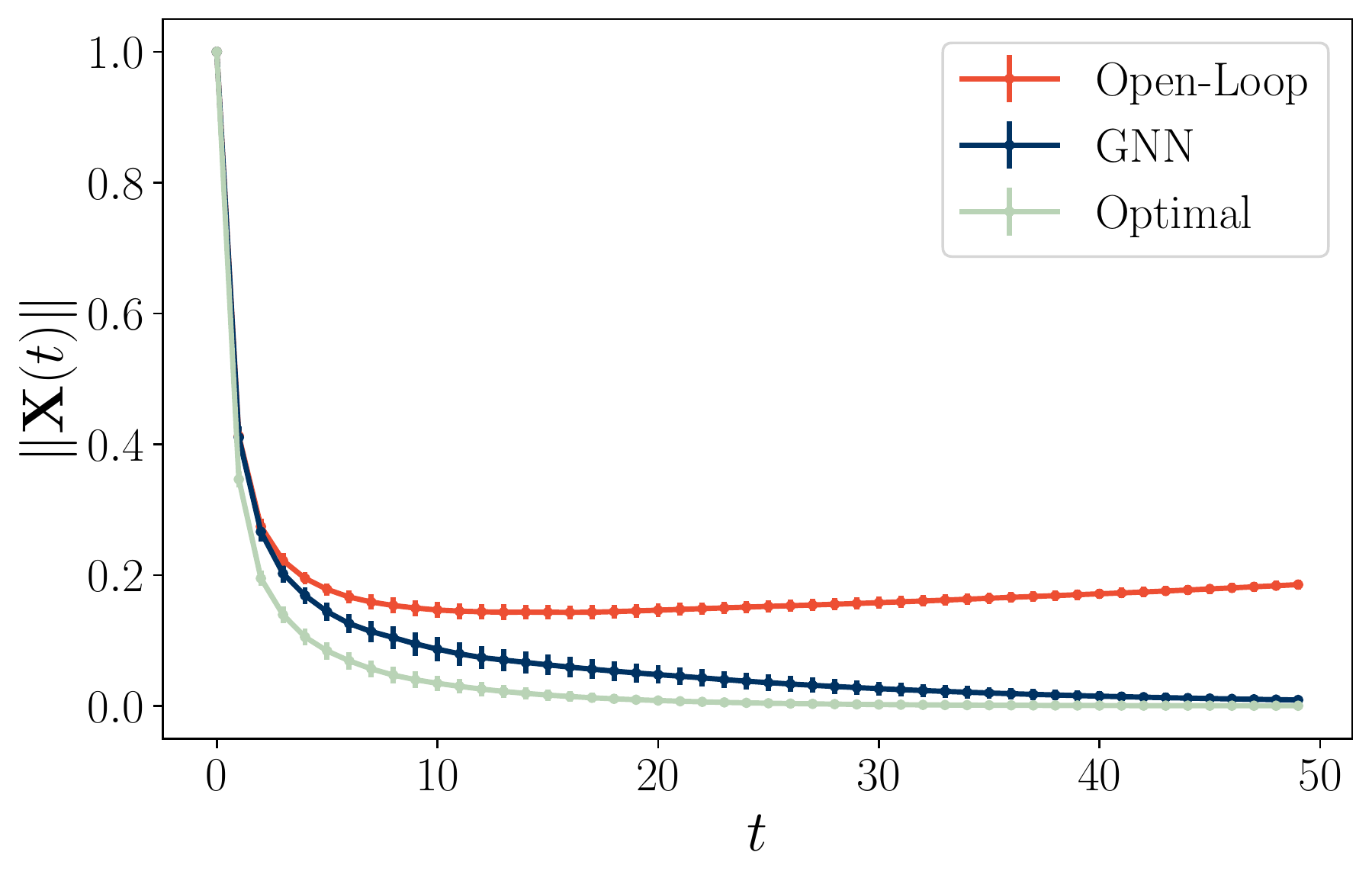}%
        \label{subfig:normUnstable}
    }
    \caption{Comparison with the open-loop system, showing the norm of the evolution of the state norm $\|\mtX(t)\|$ as a function of time $t$. (a) This is the case when the system is open-loop stable, i.e. $\|\mtA\|_{2} = 0.995$. It is observed that, while the trajectory is going to zero even in the absence of a controller (open-loop), the use of a GNN-based controller drives the state faster to zero. (b) Consider now an unstable open-loop system given by $\|\mtA\|_{2} = 1.01$. It is observed that the state does not go to zero in the absence of a controller, and that the GNN-based controller successfully drives the state to $0$.}
    \label{fig:normOpenLoop}
\end{figure*}
%%
%%%%          End of FIGURE         %%%%
%%%%%%%%%%%%%%%%%%%%%%%%%%%%%%%%%%%%%%%%

\medskip\noindent \textbf{Experiment 3: Comparison with open-loop systems.} In the third experiment, a comparison with an open-loop system is carried out. It is noted that, from choosing $\|\mtA\|_{2} = 0.995$, the resulting system is open-loop stable and, thus, the state will be driven to zero even in the absence of a controller. In this context, the effect of the distributed controller should be such that it drives the states to zero faster than the open-loop case. The results shown in Fig.~\ref{subfig:normStable} indicate that the use of a GNN controller drives the state to zero faster than the open-loop, uncontrolled, system. This illustrates that the GNN controller is better than using no controller, also in the case where the open-loop system is already stable. This is also shown in the resulting cost, which for the open-loop system is $1.5961 (\pm 0.0837)$ while for the GNN controller is $1.1104 (\pm 0.0334)$.

Alternatively, the case of a system that is open-loop unstable is also considered. In this case, the norm of the system matrix is $\|\mtA\|_{2} = 1.01$. It is immediately observed in Fig.~\ref{subfig:normUnstable} that while the open-loop system tends to be unstable (the norm of $\|\mtX(t)\|$ grows as $t$ grows), the GNN controller effectively drives the state to zero.

More generally, an experiment of the normalized cost as a function of $\|\mtA\|_{2}$ is run. This experiment helps visualize the transition between systems that are open-loop stable and systems that are not. The norm of the system matrix $\|\mtA\|_{2}$ varies from $0.95$ to $1.01$. Results are shown in Fig.~\ref{fig:costOpenLoop}. It is evident that as $\|\mtA\|_{2}$ grows, the cost increases, showing that the system is increasingly harder to control. But, while the open-loop system cost seems to exponentially grow, the GNN controller manages to keep the cost low and, as seen in Fig.~\ref{subfig:normUnstable} it effectively drives the state to zero.

%%%%%%%%%%%%%%%%%%%%%%%%%%%%%%%%%%%%%%%%
%%%%             FIGURE             %%%%  fig:costOpenLoop
%%%%%%%%%%%%%%%%%%%%%%%%%%%%%%%%%%%%%%%%
%%
\begin{figure*}[!t]
    \centering
    \includegraphics[width=0.6\textwidth]{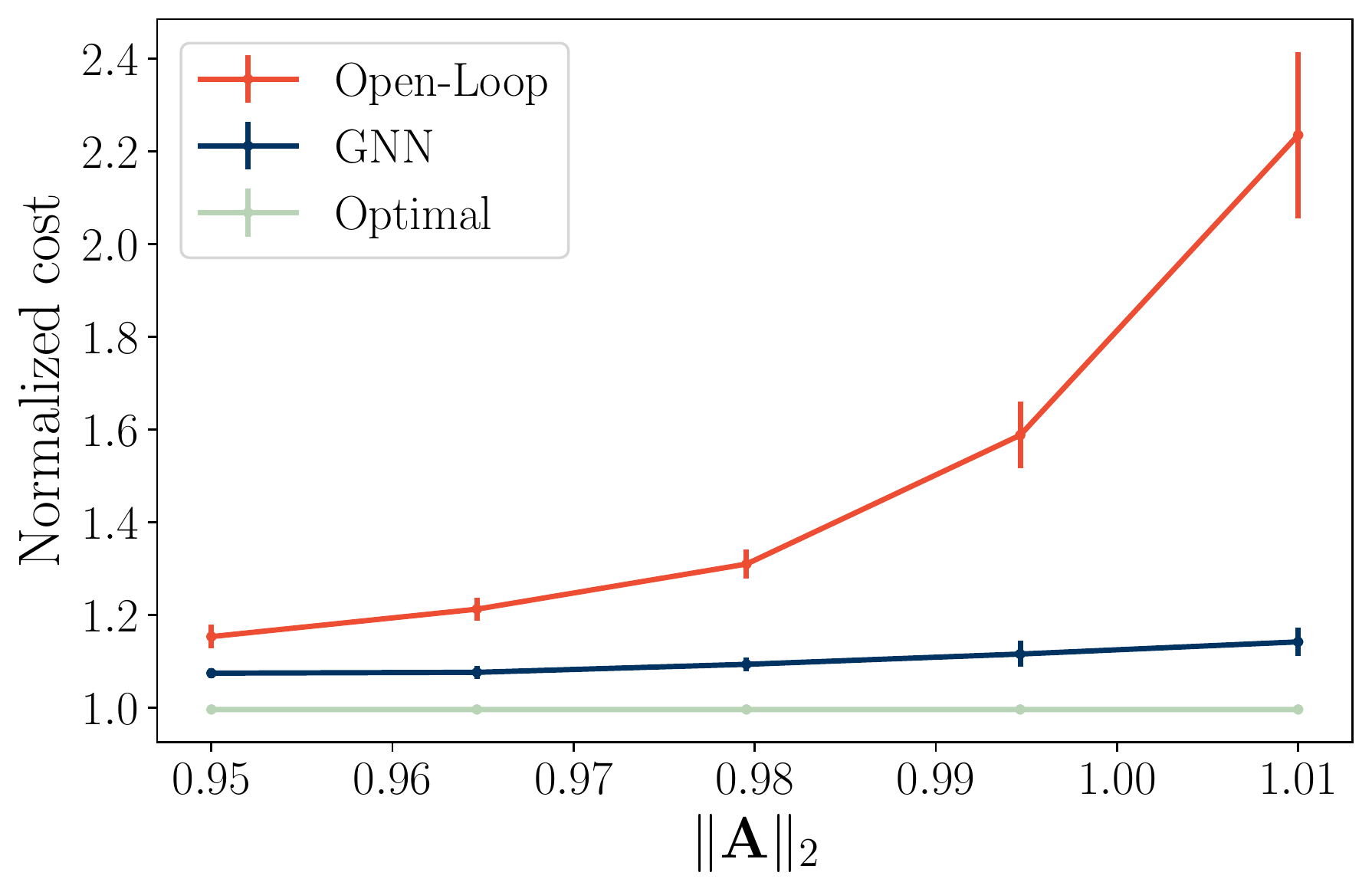}%
    \caption{Normalized cost as a function of the norm of $\|\mtA\|_{2}$. It is observed that the cost for the uncontrolled, open-loop system, grows exponentially with the norm of $\|\mtA\|_{2}$ as expected. The cost of the GNN-controller, however, grows only slightly with increasing values of $\|\mtA\|_{2}$.}
    \label{fig:costOpenLoop}
\end{figure*}
%%
%%%%          End of FIGURE         %%%%
%%%%%%%%%%%%%%%%%%%%%%%%%%%%%%%%%%%%%%%%

%%%%%%%%%%%%%%%%%%%%%%%%%%%%%%%%%%%%%%%%%%%%%%%%%%%%%%%%%%%%%%%%%%%%%%%%%%%%%%%%
%%%%                             UNKNOWN SYSTEM                             %%%%
%%%%%%%%%%%%%%%%%%%%%%%%%%%%%%%%%%%%%%%%%%%%%%%%%%%%%%%%%%%%%%%%%%%%%%%%%%%%%%%%

\medskip\noindent \textbf{Experiment 4: Unknown system matrices.} In the fourth experiment, the impact of an unknown system on both the stability (Prop.~\ref{prop:stabilityChange}) and the trajectory deviation (Thm.~\ref{thm:robust}) is studied. The controllers are trained on a system $\stD$, and then tested on another system $\sthD$ that is a random Gaussian noise perturbation such that $\fnd(\stD, \sthD) = \sceps$ for some predefined value of $\sceps$. It is observed in \eqref{eq:stabilityChange} that the change in stability is controlled by $C_{\fnPhi}=C_{\fnH_{1}}C_{\fnH_{2}}$, while \eqref{eq:robust} shows that the trajectory deviation can be controlled by lowering the value of the Lipschitz constants $\{\scGamma_{\fnH_{1}},\scGamma_{\fnH_{2}}\}$ and of the size $\{C_{\fnH_{1}}, C_{\fnH_{2}}\}$ of the filters involved. Therefore, the controller (iv: GNN) is trained with three different penalties: a penalty on the size $C_{\fnPhi}$, i.e. the objective function is $\fnJ ( \{\mtX(t)\}, \{\mtU(t)\}) + C_{\fnPhi}$, a penalty on the Lipschitz constants, i.e. $\fnJ ( \{\mtX(t)\}, \{\mtU(t)\}) +  (\scGamma_{\fnH_{1}}+ \scGamma_{\fnH_{2}})$, or a penalty on both the filter size and the Lipschitz constant, i.e. $\fnJ ( \{\mtX(t)\}, \{\mtU(t)\}) + 0.5 (\scGamma_{\fnH_{1}} + \scGamma_{\fnH_{2}} + C_{\fnPhi})$. This is indicated by the legend `GNN w/ size', `GNN w/ Lipschitz', and `GNN w/ both', respectively. The GCNN is also trained without penalties, for comparison, and labeled `GNN'.

The results are shown on Fig.~\ref{fig:unknownSystem}. First, the effects of the unknown system on the stability are analyzed, see Prop.~\ref{prop:stabilityChange}. Fig.~\ref{subfig:stability} shows that when training the GNN with a size penalty, the controller leads to a stable closed-loop system $100\%$ of the time for $\sceps < 0.05$, fails to control only $0.5\%$ of the trajectories for $\sceps = 0.0562$ and $10\%$ of the trajectories for $\sceps = 0.1$. When training with both penalties, the controller is able to lead to stable systems $100\%$ of the time for $\sceps = 0.01$, but then decays rapidly in its ability to stabilize the system as $\sceps$ grows. Training with Lipschitz penalty only leads to a controller that can stabilize about $92\%$ of the trajectories for $\sceps = 0.01$ and then falls to stabilizing about $80\%$ of the trajectories for $\sceps = 0.1$. This shows that training with a penalty on the size $C_{\fnPhi}$ of the GNN has the most impact on the ability of the learned distributed controller to stabilize the system, as predicted by Prop.~\ref{prop:stabilityChange} Finally, note that when training the GNN without penalties, the resulting controller stabilizes only $55\%$ of the trajectories on an unknown system.

%%%%%%%%%%%%%%%%%%%%%%%%%%%%%%%%%%%%%%%%
%%%%             FIGURE             %%%%  fig:unknownSystem
%%%%%%%%%%%%%%%%%%%%%%%%%%%%%%%%%%%%%%%%
%% {subfig:stability, subfig:trajectoryDeviation}
\begin{figure*}[!t]
    \centering
    \subfloat[Stability]{
        \includegraphics[width=0.475\textwidth]{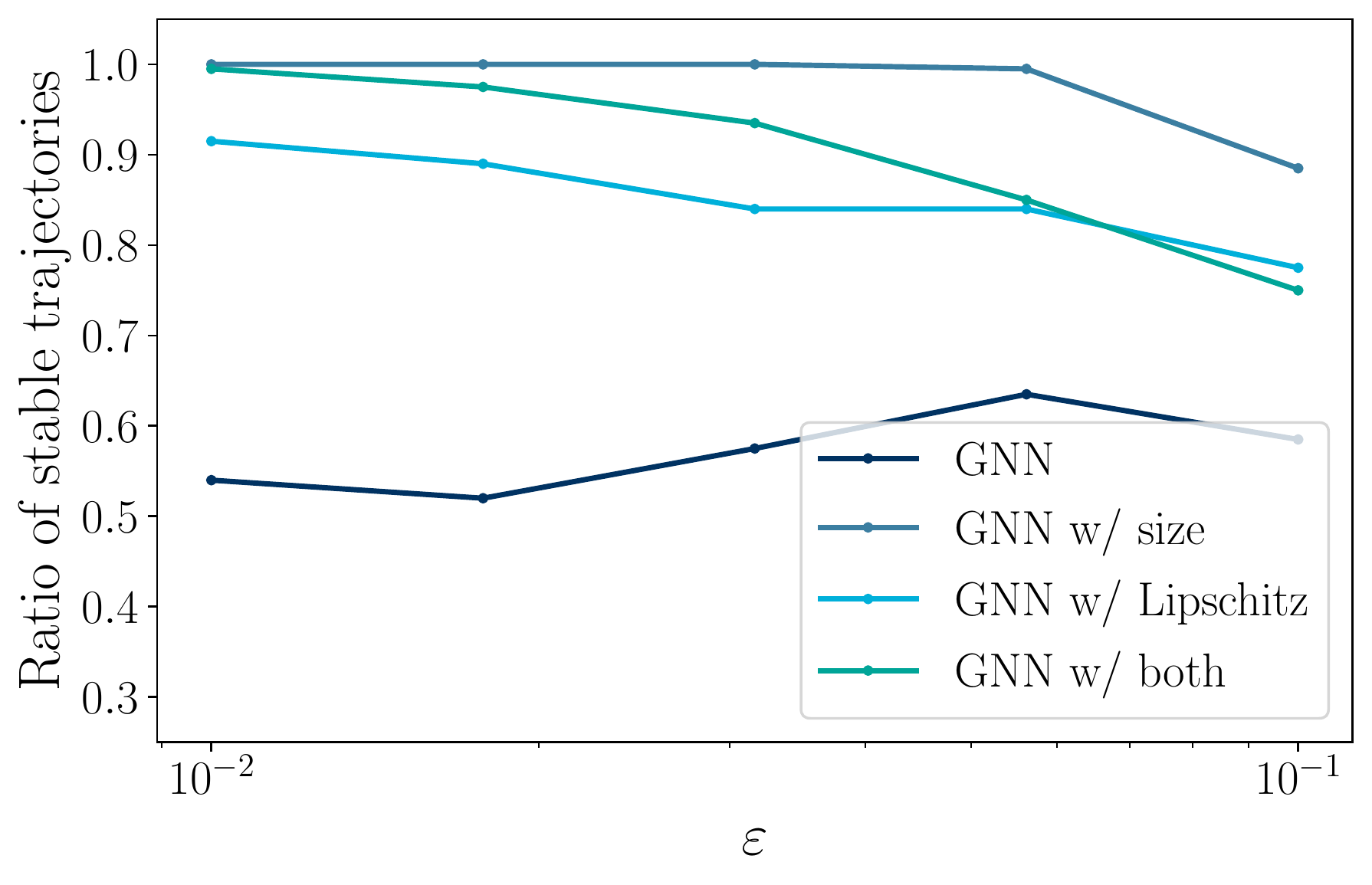}%
        \label{subfig:stability}
    }
    \hfil
    \subfloat[Trajectory deviation]{
        \includegraphics[width=0.475\textwidth]{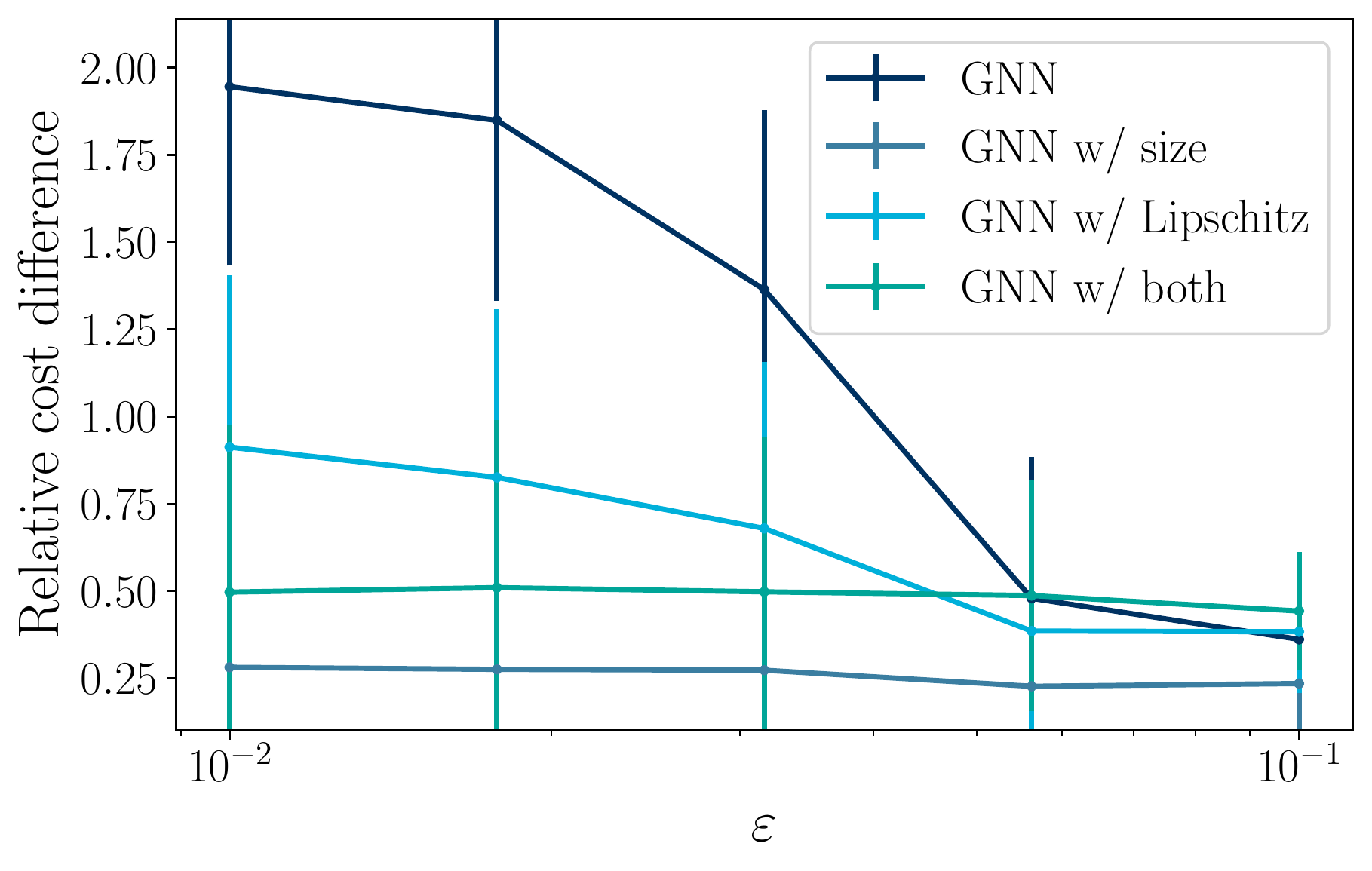}%
        \label{subfig:trajectoryDeviation}
    }
    \caption{Simulation results for a network with unknown system matrices as a function of the distance $\sceps$ between the systems, see \eqref{eq:systDistance}. (a) Ratio of stable trajectories as a function of $\sceps$; it is observed that when training with a penalty on the size $C_{\fnPhi}$ of the GNN, the resulting trajectories are stable for larger values of $\sceps$. (b) Cost difference of the controlled trajectories relative to the cost on the perfectly known system; it is observed that when training with a penalty on the size $C_{\fnPhi}$ of the GNN, the resulting controller achieves the lowest relative cost difference. The distributed controller (iii: D-MLP) and the centralized controller (ii: MLP) are not shown since they exhibit relative cost differences of approximately $7.5$ and $1400$, respectively, thus being out of scale; this is likely to their failure to control trajectories.}
    \label{fig:unknownSystem}
\end{figure*}
%%
%%%%          End of FIGURE         %%%%
%%%%%%%%%%%%%%%%%%%%%%%%%%%%%%%%%%%%%%%%

It is observed in Fig.~\ref{subfig:trajectoryDeviation} the relative difference between the cost obtained when testing on the system $\stD$ and that obtained when testing on system $\sthD$ for different values of system distance $\sceps$ among stable trajectories. First, it is noted that training with a penalty on the size of the GNN leads to a controller that is unaffected by changes in the system, exhibiting a relative cost difference of $0.25$ for all values of $\sceps$ under study. The other three controllers seem to improve in their relative difference as $\sceps$ grows, and this can be explained because the cost is being computed only among stable trajectories. This implies that, while $\sceps$ grows and less trajectories are being stabilized, the ones that remain do achieve good relative cost difference. Finally, it is noted that the distributed controller (iii: D-MLP) and the centralized learnable controller (ii: MLP) were also considered in this simulation. These controllers exhibited relative differences of approximately $7.5$ and $1400$, respectively, thus falling out of scale and not being shown in the figures. This results show that neither the (iii: D-MLP) nor the (ii: MLP) controllers are robust to changes in the system dynamics.

%%%%%%%%%%%%%%%%%%%%%%%%%%%%%%%%%%%%%%%%%%%%%%%%%%%%%%%%%%%%%%%%%%%%%%%%%%%%%%%%
%%%%                              SCALABILITY                               %%%%
%%%%%%%%%%%%%%%%%%%%%%%%%%%%%%%%%%%%%%%%%%%%%%%%%%%%%%%%%%%%%%%%%%%%%%%%%%%%%%%%

%%%%%%%%%%%%%%%%%%%%%%%%%%%%%%%%%%%%%%%%
%%%%             FIGURE             %%%%  fig:scalability
%%%%%%%%%%%%%%%%%%%%%%%%%%%%%%%%%%%%%%%%
%%
\begin{figure*}[!t]
    \centering
    \includegraphics[width=0.6\textwidth]{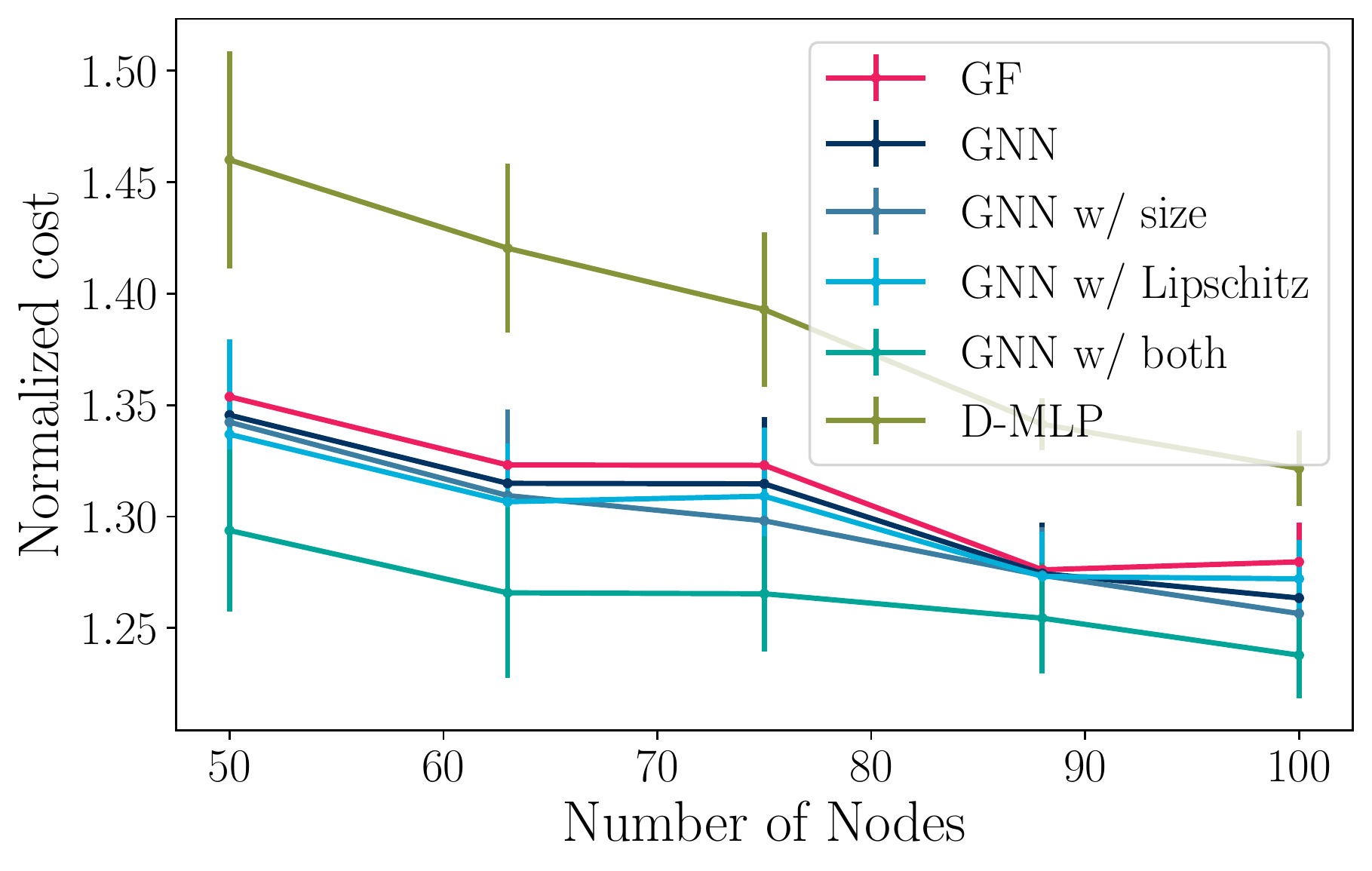}%
    \caption{Normalized cost for the stable trajectories of a GNN-based controller trained on $50$ nodes and tested on a larger network system. It is observed that training with penalties on both the Lipschitz constant and the size of the filters lead to best scalability results.}
    \label{fig:scalability}
\end{figure*}
%%
%%%%          End of FIGURE         %%%%
%%%%%%%%%%%%%%%%%%%%%%%%%%%%%%%%%%%%%%%%

\medskip\noindent \textbf{Experiment 5: Scalability.} In the last experiment, scalability of the distributed controllers (iii)-(v) is compared. These methods are trained on a system with $N=50$ nodes, and then at test time, are used on increasingly larger systems $N \in \{50, 63, 75, 87, 100\}$. The resulting costs of the stable trajectories are shown in Fig.~\ref{fig:scalability}. It is observed that, while the D-MLP performs better when tested on the same system as it was trained (see experiment~2), it does not transfer as well to larger systems. This is  likely to be because it assigns a different fully connected neural network controller to each component, so that, when tested on larger systems, it has to replicate this controller on other nodes and that may have a substantially different topological neighborhood. Controllers (iv: GNN) and (v: GF), on the other hand, successfully adapt to larger systems, even when trained on small ones. In particular, training with penalties on both the Lipschitz constant and the size of the filters leads to the best scalability results. It is noted that the centralized controller (ii: MLP) cannot transfer to systems with different number of nodes since the number of learned parameters depends on the number of nodes.

%%%%%%%%%%%%%%%%%%%%%%%%%%%%%%%%%%%%%%%%%%%%%%%%%%%%%%%%%%%%%%%%%%%%%%%%%%%%%%%%
%%%%                                                                        %%%%
%%%%                              CONCLUSIONS                               %%%%
%%%%                                                                        %%%%
%%%%%%%%%%%%%%%%%%%%%%%%%%%%%%%%%%%%%%%%%%%%%%%%%%%%%%%%%%%%%%%%%%%%%%%%%%%%%%%%

\section{Conclusion} \label{sec:conclusions}

%!TEX root = 00-distributedLQR.tex

%%%%%%%%%%%%%%%%%%%%%%%%%%%%%%%%%%%%%%%%%%%%%%%%%%%%%%%%%%%%%%%%%%%%%%%%%%%%%%%%
%%%%                                                                        %%%%
%%%%                              CONCLUSIONS                               %%%%
%%%%                                                                        %%%%
%%%%%%%%%%%%%%%%%%%%%%%%%%%%%%%%%%%%%%%%%%%%%%%%%%%%%%%%%%%%%%%%%%%%%%%%%%%%%%%%
%%%% sec:conclusions
%%%%%%%%%%%%%%%%%%%%

This paper proposes to address the issue of the intractability of distributed optimal controllers by leveraging a nonlinear GNN-based parametrization. While the resulting controller is suboptimal, it exhibits several desirable properties such as distributed computation, efficiency and scalability. These controllers are applied to the distributed linear-quadratic problem, which can be cast as a self-supervised empirical risk minimization problem, and then solved by means of machine learning techniques. A sufficient condition for the resulting closed-loop system to be input-state stable is derived in terms of the filter taps of the GNN-based controller. Additionally, the trajectory deviation due to mismatch of the system descriptions is shown to also be controlled by the filter taps. Extensive simulations illustrate the satisfactory performance exhibited by GNN-based controllers as well as the ability to be trained to exhibit certain desirable characteristics such as an improved closed-loop stability or a smaller trajectory deviation under model mismatch. The resulting controller is also shown to scale to larger systems. Future research on the topic may involve the study of equilibrium points of a GNN-controlled system and their Lyapunov stability, the use of distributed optimization techniques to solve the self-supervised learning problem, and the adoption of other non-convolutional GNN-based architectures.

%\gray{Mention that the current controller leverages only present information. That looking into a controller that can also leverage past information may improve performance and is left as future research.}

%%%%%%%%%%%%%%%%%%%%%%%%%%%%%%%%%%%%%%%%%%%%%%%%%%%%%%%%%%%%%%%%%%%%%%%%%%%%%%%%
%%%%                                                                        %%%%
%%%%                                APPENDIX                                %%%%
%%%%                                                                        %%%%
%%%%%%%%%%%%%%%%%%%%%%%%%%%%%%%%%%%%%%%%%%%%%%%%%%%%%%%%%%%%%%%%%%%%%%%%%%%%%%%%

% or



\appendix

%!TEX root = 00-distributedLQR.tex

%%%%%%%%%%%%%%%%%%%%%%%%%%%%%%%%%%%%%%%%%%%%%%%%%%%%%%%%%%%%%%%%%%%%%%%%%%%%%%%%
%%%%                                                                        %%%%
%%%%                           AUXILIARY RESULTS                            %%%%
%%%%                                                                        %%%%
%%%%%%%%%%%%%%%%%%%%%%%%%%%%%%%%%%%%%%%%%%%%%%%%%%%%%%%%%%%%%%%%%%%%%%%%%%%%%%%%
%%%% app:aux
%%%%%%%%%%%%

\section{Auxiliary Results} \label{app:aux}

In this appendix four Lemmas that are useful for proving the theorems and propositions of Sections~\ref{app:stability}~and~\ref{app:robust} are included. The first two Lemmas establish an upper bound on the output of a graph filter (Lemma~\ref{lemma:boundFilter}) and a GNN (Lemma~\ref{lemma:boundGCNN}) as a function of the size of the filters involved. The following two lemmas determine the Lipschitz continuity with respect to the support matrix $\mtS$ of the graph filter (Lemma~\ref{lemma:lipschitzFilter}) and the GNN (Lemma~\ref{lemma:lipschitzGCNN}) as a function of the filter sizes and the Lipschitz constants.

%%%%%%%%%%%%%%%%%%%%%%%%%%%%%%%%%%%%%%%%%%%%%%%%%%%%%%%%%%%%%%%%%%%%%%%%%%%%%%%%
%%%%                         LEMMA: Bound on filter                         %%%%
%%%%%%%%%%%%%%%%%%%%%%%%%%%%%%%%%%%%%%%%%%%%%%%%%%%%%%%%%%%%%%%%%%%%%%%%%%%%%%%%
%%%% lemma:boundFilter
%%%%%%%%%%%

\begin{lemma}[Bound on Graph Filter Output]
    \label{lemma:boundFilter}
    Let $\fnH: \fdR^{N \times F} \to \fdR^{N \times G}$ be a graph filter \eqref{eq:graphFilter} defined over a support matrix $\mtS \in \fdR^{N \times N}$. Let $\mtX \in \fdR^{N \times F}$ be any graph signal such that $\|\mtX\| < \infty$. Then,
    % eq:boundFilter
    \begin{equation} \label{eq:boundFilter}
        \big\| \fnH(\mtX;\mtS,\stH)\big\| \leq C_{\fnH} \big\| \mtX\big\|,
    \end{equation}
    with $C_{\fnH}$ being the size of the filter bank, see \eqref{eq:graphFilterNorm}.
\end{lemma}

%%%%%%%%%%%%%%%%%%%%%%%%%%%%%%%%%%%%%%%%
%%%%             PROOF              %%%%
%%%%%%%%%%%%%%%%%%%%%%%%%%%%%%%%%%%%%%%%

\begin{proof}
    Recall that the norm associated to the graph signal space is given by the $L_{2,1}$ entrywise matrix norm, see \eqref{eq:graphSignalNorm}. Then, the the graph signal size of the output $\mtY = \fnH(\mtX; \mtS, \stH)$ can be computed as
    % eq:signalNormOut
    \begin{equation} \label{eq:signalNormOut}
        \| \mtY \| = \sum_{g=1}^{G} \| \vcy^{g} \|_{2} = \sum_{g=1}^{G} \Big\| \sum_{f=1}^{F} \mtH_{fg}(\mtS) \vcx^{f} \Big\|_{2},
    \end{equation}
    where $\mtH_{fg}(\mtS) = \sum_{k=0}^{K} [\mtH_{k}]_{fg} \mtS^{k}$, see \eqref{eq:graphFilterNorm}, and where $\|\vcx\|_{2}$ represents the Euclidean norm on vectors.  One can apply the triangular inequality to \eqref{eq:signalNormOut} to obtain:
    % eq:signalNormOutTrIneq
    \begin{equation} \label{eq:signalNormOutTrIneq}
        \| \mtY \| \leq \sum_{g=1}^{G} \sum_{f=1}^{F} \big\| \mtH_{fg}(\mtS) \vcx^{f} \big\|_{2}
    \end{equation}
    and noticing that the summation is comprised of Euclidean vector norms, the submultiplicativity of the corresponding matrix spectral norm can be used to arrive at
    % eq:
    \begin{equation}
        \|\mtY \| \leq \sum_{g=1}^{G} \sum_{f=1}^{F} \big\| \mtH_{fg}(\mtS) \big\|_{2} \big\| \vcx^{f} \big\|_{2},
    \end{equation}
    which, noting that the sum over $g$ only affects $\|\mtH_{fg}(\mtS)\|_{2}$, can be rearranged as
    % eq:signalNormOutSub
    \begin{equation} \label{eq:signalNormOutSub}
        \|\mtY \| \leq \sum_{f=1}^{F}  \big\| \vcx^{f} \big\|_{2} \sum_{g=1}^{G} \big\| \mtH_{fg}(\mtS) \big\|_{2}.
    \end{equation}
    Next, note that $\sum_{g=1}^{G} \|\mtH_{fg}(\mtS)\|_{2}$ is the sum of all the spectral norms of the filters along the $g$ dimension, thus the result is a scalar that depends on $f$ and is denoted with $C_{f}$ in this proof, i.e. $\sum_{g=1}^{G} \|\mtH_{fg}(\mtS)\|_{2}=C_{f}$. For each value of $f$, there is a different $C_{f}$, and it holds true that $C_{f} \leq \sup_{f=1,\ldots,F}$. This implies that $\sum_{g=1}^{G} \| \mtH_{fg}(\mtS) \|_{2} \leq \sup_{f=1,\ldots,F} \sum_{g=1}^{G} \| \mtH_{fg}(\mtS)\|_{2}$.

    From \eqref{eq:graphFilterNorm}, note that each element of the matrix $\mtC_{\fnH} \in \fdR^{F \times G}$ is given by $\max_{\lambda \in [\lambda_{l},\lambda_{h}]} |h_{fg}(\lambda)|$ for some chosen values of $[\lambda_{l},\lambda_{h}]$. Then, if $\lambda_{l}$ and $\lambda_{h}$ are the minimum and maximum eigenvalues of $\mtS$ as is usually the case, then it follows that $\sup_{f=1,\ldots,F} \sum_{g=1}^{G} \| \mtH_{fg}(\mtS)\|_{2} \leq \| \mtC_{\fnH}\|_{\infty} = C_{\fnH}$, see \eqref{eq:graphFilterNorm}. Recall that $\|\mtA\|_{\infty}$ is the infinity norm of matrices (i.e. maximum absolute row sum). Finally,  \eqref{eq:signalNormOutSub} can be upper bounded as
    % eq:signalNormOutBg
    \begin{equation} \label{eq:signalNormOutBg}
        \|\mtY \| \leq C_{\fnH} \sum_{f=1}^{F} \| \vcx^{f} \|_{2}.
    \end{equation}
    Noting that $\sum_{f=1}^{F} \|\vcx^{f}\|_{2} = \|\mtX\|$ completes the proof.
\end{proof}

%%%%%%%%%%%%%%%%%%%%%%%%%%%%%%%%%%%%%%%%%%%%%%%%%%%%%%%%%%%%%%%%%%%%%%%%%%%%%%%%
%%%%                          LEMMA: Bound on GCNN                          %%%%
%%%%%%%%%%%%%%%%%%%%%%%%%%%%%%%%%%%%%%%%%%%%%%%%%%%%%%%%%%%%%%%%%%%%%%%%%%%%%%%%
%%%% lemma:boundGCNN
%%%%%%%%%%%

\begin{lemma}[Bound on GNN Output]
    \label{lemma:boundGCNN}
    Let $\fnPhi(\cdot;\mtS,\stH): \fdR^{N \times F} \to \fdR^{N \times G}$ be a GNN \eqref{eq:GCNN} with $L$ layers defined over a support matrix $\mtS \in \fdR^{N \times N}$. Let the nonlinearity $\fnsigma(\cdot)$ be such that $|\fnsigma(x)| \leq C_{\sigma} |x|$ for all $x \in \fdR$, for some $C_{\sigma} > 0$. Then, for every graph signal $\mtX \in \fdR^{N \times F}$ with $\|\mtX\| < \infty$, it holds that
    % eq:boundGCNN
    \begin{equation} \label{eq:boundGCNN}
        \big\| \fnPhi(\mtX;\mtS,\stH)\big\| \leq C_{\sigma}^{L}C_{\fnPhi} \big\| \mtX\big\|,
    \end{equation}
    where $C_{\fnPhi}=\prod_{\ell=1}^{L} C_{\fnH_{\ell}}$ for $C_{\fnH_{\ell}}$ the size of the $\ell^{\text{th}}$ filter, see \eqref{eq:graphFilterNorm}.
\end{lemma}

%%%%%%%%%%%%%%%%%%%%%%%%%%%%%%%%%%%%%%%%
%%%%             PROOF              %%%%
%%%%%%%%%%%%%%%%%%%%%%%%%%%%%%%%%%%%%%%%

\begin{proof}
    Consider the computation of layer $\ell$
    % eq:GCNNlayer
    \begin{equation} \label{eq:GCNNlayer}
        \mtX_{\ell} = \fnsigma \Big( \fnH_{\ell} \big( \mtX_{\ell-1};\mtS,\stH_{\ell} \big) \Big),
    \end{equation}
    whose norm is given by \eqref{eq:graphSignalNorm},
    % eq:normGCNNlayer
    \begin{equation} \label{eq:normGCNNlayer}
        \| \mtX_{\ell} \| = \sum_{g=1}^{F_{\ell}} \| \vcx_{\ell}^{g} \|_{2},
    \end{equation}
    with
    % eq:GCNNlayerSingle
    \begin{equation} \label{eq:GCNNlayerSingle}
        \vcx_{\ell}^{g} = \fnsigma \Big( \sum_{f=1}^{F_{\ell-1}} \mtH_{\ell fg}(\mtS) \vcx_{\ell-1}^{f} \Big),
    \end{equation}
    where $\mtH_{\ell fg}(\mtS) = \sum_{k=0}^{K_{\ell}} [\mtH_{\ell k}]_{fg} \mtS^{k}$ denotes the scalar-valued graph filter.

    Substituting \eqref{eq:GCNNlayerSingle} into \eqref{eq:normGCNNlayer} and using the hypothesis on the nonlinearity that $|\fnsigma(x)| \leq C_{\sigma}|x|$ for all $x$, the following upper bound on the norm of the output signal at layer $\ell$ is obtained:
    % eq:
    \begin{equation}
        \|\mtX_{\ell}\| \leq C_{\sigma} \sum_{g=1}^{F_{\ell}} \Big\| \sum_{f=1}^{F_{\ell-1}} \mtH_{\ell fg}(\mtS) \vcx_{\ell-1}^{f} \Big\|_{2},
    \end{equation}
    which is simply
    % eq:boundGCNNlayerFilter
    \begin{equation} \label{eq:boundGCNNlayerFilter}
        \|\mtX_{\ell}\| \leq C_{\sigma} \big\| \fnH_{\ell}(\mtX_{\ell-1};\mtS,\stH_{\ell}) \big\|.
    \end{equation}
    Now, using Lemma~\ref{lemma:boundFilter} on \eqref{eq:boundGCNNlayerFilter} yields
    % eq:boundGCNNlayerFilterComplete
    \begin{equation} \label{eq:boundGCNNlayerFilterComplete}
        \|\mtX_{\ell}\| \leq C_{\sigma} C_{\fnH_{\ell}} \|\mtX_{\ell-1} \|.
    \end{equation}

    Repeating \eqref{eq:boundGCNNlayerFilterComplete} for all consecutive layers until reaching $\ell=1$ leads to
    % eq:boundGCNNlayerFilterInitial
    \begin{equation} \label{eq:boundGCNNlayerFilterInitial}
    \|\mtX_{\ell}\| \leq C_{\sigma}^{\ell} \prod_{\ell'=1}^{\ell} C_{\fnH_{\ell'}} \|\mtX_{0} \|.
    \end{equation}
    By substituting $\ell=L$ into \eqref{eq:boundGCNNlayerFilterInitial} and recalling that $\mtX_{0} = \mtX$, $\fnPhi(\mtX;\mtS,\stH) = \mtX_{L}$ and $C_{\fnPhi} = \prod_{\ell=1}^{L} C_{\fnH_{\ell}}$, the proof is completed.
\end{proof}

In what follows, we state two Lemmas regarding the Lipschitz continuity of graph filters and GNNs with respect to the support matrix $\mtS$. These results have already been correspondingly proved, and are just rewritten here to unify notation.

%%%%%%%%%%%%%%%%%%%%%%%%%%%%%%%%%%%%%%%%%%%%%%%%%%%%%%%%%%%%%%%%%%%%%%%%%%%%%%%%
%%%%                 LEMMA: Lipschitz continuity of filter                  %%%%
%%%%%%%%%%%%%%%%%%%%%%%%%%%%%%%%%%%%%%%%%%%%%%%%%%%%%%%%%%%%%%%%%%%%%%%%%%%%%%%%
%%%% lemma:lipschitzFilter
%%%%%%%%%%%

\begin{lemma}[Lipschitz continuity of graph filter with respect to $\mtS$]
    \label{lemma:lipschitzFilter}
    Let $\fnH: \fdR^{N \times F} \to \fdR^{N \times G}$ be a graph filter \eqref{eq:graphFilter}. Let $\mtS \in \fdR^{N \times N}$ and $\mthS \in \fdR^{N \times N}$ be two support matrices, such that $\| \mtS - \mthS\|_{2} \leq \sceps$. Then, for any graph signal $\mtX \in \fdR^{N \times F}$ such that $\|\mtX\| < \infty$, it holds that
    % eq:lipschitzFilter
    \begin{equation} \label{eq:lipschitzFilter}
    \big\| \fnH(\mtX;\mthS,\stH) - \fnH(\mtX;\mtS,\stH) \big\| \leq \sceps(1+ 8 \sqrt{N})\scGamma_{\fnH} \| \mtX\| + \bigOh(\sceps^{2}),
    \end{equation}
    with $\scGamma_{\fnH}$ being the Lipschitz constant filter bank, see \eqref{eq:LipschitzFilters}.
\end{lemma}

%%%%%%%%%%%%%%%%%%%%%%%%%%%%%%%%%%%%%%%%
%%%%             PROOF              %%%%
%%%%%%%%%%%%%%%%%%%%%%%%%%%%%%%%%%%%%%%%

\begin{proof}
See \cite[Thm. 1]{Gama20-Stability}.
\end{proof}

%%%%%%%%%%%%%%%%%%%%%%%%%%%%%%%%%%%%%%%%%%%%%%%%%%%%%%%%%%%%%%%%%%%%%%%%%%%%%%%%
%%%%                  LEMMA: Lipschitz continuity of GCNN                   %%%%
%%%%%%%%%%%%%%%%%%%%%%%%%%%%%%%%%%%%%%%%%%%%%%%%%%%%%%%%%%%%%%%%%%%%%%%%%%%%%%%%
%%%% lemma:lipschitzGCNN
%%%%%%%%%%%

\begin{lemma}[Lipschitz continuity of the GNN with respect to $\mtS$]
    \label{lemma:lipschitzGCNN}
    Let $\fnPhi(\cdot;\cdot,\stH): \fdR^{N \times F} \to \fdR^{N \times G}$ be a GNN \eqref{eq:GCNN} with $L$ layers. Let $\fnsigma(\cdot)$ be such that $|\fnsigma(x) - \fnsigma(y)| \leq \scGamma_{\sigma} |x-y|$ for all $x,y \in \fdR$ for some $\scGamma_{\fnsigma} > 0$, and $\fnsigma(0) = 0$. Let $\mtS \in \fdR^{N \times N}$ and $\mthS \in \fdR^{N \times N}$ be two support matrices such that $\| \mtS - \mthS \|_{2} \leq \sceps$. Then, for every graph signal $\mtX \in \fdR^{N \times F}$ with $\|\mtX\| < \infty$, it holds that
    % eq:lipschitzGCNN
    \begin{equation} \label{eq:lipschitzGCNN}
    \big\| \fnPhi(\mtX;\mthS,\stH) - \fnPhi(\mtX;\mtS,\stH)\big\| \leq\sceps (1+8 \sqrt{N}) \scGamma_{\fnsigma}^{L} C_{\fnPhi} \sum_{\ell=1}^{L} \frac{\scGamma_{\fnH_{\ell}}}{C_{\fnH_{\ell}}} \| \mtX \| + \bigOh(\sceps^{2}),
    \end{equation}
    where $C_{\fnPhi} = \prod_{\ell=1}^{L} C_{\fnH_{\ell}}$ for $C_{\fnH_{\ell}}$ the size of $\ell^{\text{th}}$ filter, see \eqref{eq:graphFilterNorm}, and where $\scGamma_{\fnH_{\ell}}$ is the corresponding Lipschitz constant, see \eqref{eq:graphFilterLipzchitz}.
\end{lemma}

%%%%%%%%%%%%%%%%%%%%%%%%%%%%%%%%%%%%%%%%
%%%%             PROOF              %%%%
%%%%%%%%%%%%%%%%%%%%%%%%%%%%%%%%%%%%%%%%

\begin{proof}
    See \cite[Thm. 4]{Gama20-Stability}.
\end{proof}

%!TEX root = 00-distributedLQR.tex

%%%%%%%%%%%%%%%%%%%%%%%%%%%%%%%%%%%%%%%%%%%%%%%%%%%%%%%%%%%%%%%%%%%%%%%%%%%%%%%%
%%%%                                                                        %%%%
%%%%                     PROOF OF CLOSED-LOOP STABILITY                     %%%%
%%%%                                                                        %%%%
%%%%%%%%%%%%%%%%%%%%%%%%%%%%%%%%%%%%%%%%%%%%%%%%%%%%%%%%%%%%%%%%%%%%%%%%%%%%%%%%
%%%% app:stability
%%%%%%%%%%%%%%%%%%

\section{Proof of Closed-Loop Stability} \label{app:stability}

In this appendix, we first prove Theorem~\ref{thm:stability} that gives a sufficient condition for the GNN-controlled system $\stD$ to be stable. We then prove Proposition~\ref{prop:stabilityChange} stating how the stability constant $\scxi$ changes from system $\stD$ to system $\sthD$.

%%%%%%%%%%%%%%%%%%%%%%%%%%%%%%%%%%%%%%%%%%%%%%%%%%%%%%%%%%%%%%%%%%%%%%%%%%%%%%%%
%%%%                       PROOF OF STABILITY THEOREM                       %%%%
%%%%%%%%%%%%%%%%%%%%%%%%%%%%%%%%%%%%%%%%%%%%%%%%%%%%%%%%%%%%%%%%%%%%%%%%%%%%%%%%

\begin{proof}[Proof of Theorem~\ref{thm:stability}]
The system dynamics with a GNN-based, exploratory controller given by $\mtU(t) = \fnPhi(\mtX(t);\mtS,\stH) + \mtE(t)$ are
% eq:
\begin{equation}
    \mtX(t) = \mtA \mtX(t-1) \mtbA + \mtB \fnPhi(\mtX(t-1)) \mtbB + \mtB \mtE(t-1) \mtbB.
\end{equation}
The graph signal norm of the trajectory can be bounded by applying the triangular inequality as follows:
% eq:ineqX
\begin{align} \label{eq:ineqX}
    \| \mtX&(t)\| \leq \| \mtA \|_{2} \|\mtbA\|_{\infty} \| \mtX(t-1) \| \\ & + \| \mtB \|_{2} \|\mtbB\|_{\infty} \| \fnPhi(\mtX(t-1)) \| + \| \mtB \|_{2} \|\mtbB\|_{\infty} \| \mtE(t-1) \|. \nonumber
\end{align}
The term $\|\fnPhi(\mtX(t);\mtS,\stH)\|$ can be bounded by leveraging Lemma~\ref{lemma:boundGCNN} on the bound of the output of a GNN as
% eq:ineqU
\begin{equation} \label{eq:ineqU}
    \|\mtU(t)\| = \big\| \fnPhi \big( \mtX(t); \mtS, \stH \big) \big\| \leq C_{\fnPhi} \| \mtX(t)\|,
\end{equation}
with $C_{\fnsigma}=1$. This result is used in \eqref{eq:ineqX}, to yield
% eq:ineqAllX
\begin{equation} \label{eq:ineqAllX}
    x_{t} \leq \scxi x_{t-1} + \beta e_{t-1},
\end{equation}
where $x_{t} = \| \mtX(t)\|$, $\scxi=\|\mtA\|_{2}\|\mtbA\|_{\infty} + C_{\fnPhi} \|\mtB\|_{2} \|\mtbB\|_{\infty}$ is given in \eqref{eq:stabilityConstant},  $\beta = \| \mtB \|_{2} \|\mtbB\|_{\infty}$ and $e_{t} = \| \mtE(t)\|$. By repeatedly applying \eqref{eq:ineqAllX}, the following inequality is obtained:
% eq:ineqXinit
\begin{equation} \label{eq:ineqXinit}
    x_{t} \leq \scxi^{t} x_{0} + \beta \sum_{\tau = 0}^{t-1} \scxi^{\tau} e_{t-\tau -1}.
\end{equation}

Now, considering the summation series that defines the stability as in \eqref{eq:stability}, one obtains:
% eq:L2conditionSufficient
\begin{equation} \label{eq:L2conditionSufficient}
    \sum_{t=0}^{\infty}x_{t} \leq x_{0} \sum_{t=0}^{\infty}\scxi^{t} + \scbeta \sum_{t=0}^{\infty} \sum_{\tau=0}^{t-1} \scxi^{\tau} e_{t-\tau-1}.
\end{equation}
Leveraging the assumptions that $\scxi < 1$ and $\sum_{t=0}^{\infty} e_{t} < \infty$, the above inequality yields
% eq:stabilityConditionNoDivision
\begin{equation} \label{eq:stabilityConditionNoDivision}
    \sum_{t=0}^{\infty} x_{t} \leq \frac{x_{0}}{1-\scxi} + \frac{\beta}{1-\scxi} \sum_{t=0}^{\infty} e_{t},
\end{equation}
where the fact that, under these assumptions, it holds that $\sum_{t=0}^{\infty} \sum_{\tau=0}^{t-1} \scxi^{\tau} e_{t-\tau-1} \leq (\sum_{t=0}^{\infty} e_{t})(\sum_{t=0}^{\infty} \scxi^{t})$ was used. The proof is complete by replacing the definitions of $x_{t}$, $e_{t}$ and $\beta$ in \eqref{eq:stabilityConditionNoDivision}. Thus, the system is input-state stable with constants $\beta_{0} = \|\mtX(0)\|/(1-\scxi)$ and $\beta_{1} = \|\mtB\|_{2} \|\mtbB\|_{\infty}/(1-\scxi)$.
\end{proof}

%%%%%%%%%%%%%%%%%%%%%%%%%%%%%%%%%%%%%%%%%%%%%%%%%%%%%%%%%%%%%%%%%%%%%%%%%%%%%%%%
%%%%                      PROOF OF STABILITY DEVIATION                      %%%%
%%%%%%%%%%%%%%%%%%%%%%%%%%%%%%%%%%%%%%%%%%%%%%%%%%%%%%%%%%%%%%%%%%%%%%%%%%%%%%%%

Next, we prove the change in the stability constant when $\fnd(\stD, \sthD) = \sceps$.

\begin{proof}[Proof of Proposition~\ref{prop:stabilityChange}]
    Start by writing the stability constant $\schxi = \scxi(\sthD,\stH)$ as given by \eqref{eq:stabilityConstant} to obtain
    % eq:schxi
    \begin{equation} \label{eq:schxi}
        \schxi = \schxi(\sthD,\stH) = \| \mthA\|_{2} \| \mthbA\|_{\infty} + C_{\fnPhi} \| \mthB\|_{2} \|\mthbB\|_{\infty}.
    \end{equation}
    This equation is equivalent to
    % eq:schxiscxi
    \begin{equation} \label{eq:schxiscxi}
    \begin{aligned}
        \schxi & = \| \mthA\|_{2} \| \mthbA\|_{\infty} - \| \mtA\|_{2} \| \mthA\|_{\infty} \\ & \quad + C_{\fnPhi} \Big( \| \mthB\|_{2} \|\mthbB\|_{\infty} - \| \mtB\|_{2} \|\mtbB\|_{\infty} \Big) + \scxi.
    \end{aligned}
    \end{equation}
    The first term can be rewritten as
    \begin{align} \label{eq:boundAnorm}
        \| & \mthA\|_{2} \| \mthbA\|_{\infty} - \| \mtA\|_{2} \|  \mtbA\|_{\infty} \\
        & = \big( \| \mthA\|_{2} - \|\mtA\|_{2} \big) \|\mthbA\|_{\infty} + \|\mtA\|_{2} \big(\| \mthbA\|_{\infty} - \|\mtbA\|_{\infty}\big). \nonumber
    \end{align}
    From the definition of the distance $\fnd(\stD,\sthD) = \sceps$ it is known that $- \sceps \leq \| \mthA \|_{2} - \|\mtA\|_{2} \leq \sceps$, and analogously for $\|\mtbA\|_{\infty}$, so that \eqref{eq:boundAnorm} can be bounded by
    % eq:boundB
    \begin{equation} \label{eq:boundA}
        \| \mthA\|_{2} \| \mthbA\|_{\infty} - \| \mtA\|_{2} \| \mthA\|_{\infty}    \leq \sceps \big(\| \mtA \|_{2} + \|\mthbA\|_{\infty} \big).
    \end{equation}
    Following the same reasoning for the control matrices, one obtains:
    % eq:boundB
    \begin{equation} \label{eq:boundB}
        \| \mthB\|_{2} \| \mthbB\|_{\infty} - \| \mtB\|_{2} \|  \mthB\|_{\infty} \leq \sceps \big(\| \mtB \|_{2} + \|\mthbB\|_{\infty}  \big).
    \end{equation}
    By substituting \eqref{eq:boundA} and \eqref{eq:boundB} into \eqref{eq:schxiscxi} and defining $\schC_{\scxi} = \| \mtA\|_{2} + \| \mthbA\|_{\infty} + \scC_{\fnPhi} (\| \mtB\|_{2} + \| \mthbB\|_{\infty})$, the proof is complete.
\end{proof}

%!TEX root = 00-distributedLQR.tex

%%%%%%%%%%%%%%%%%%%%%%%%%%%%%%%%%%%%%%%%%%%%%%%%%%%%%%%%%%%%%%%%%%%%%%%%%%%%%%%%
%%%%                                                                        %%%%
%%%%                          PROOF OF ROBUSTNESS                           %%%%
%%%%                                                                        %%%%
%%%%%%%%%%%%%%%%%%%%%%%%%%%%%%%%%%%%%%%%%%%%%%%%%%%%%%%%%%%%%%%%%%%%%%%%%%%%%%%%
%%%% app:robust
%%%%%%%%%%%%%%%

\section{Proof of Trajectory Deviations} \label{app:robust}

In this appendix, Theorem~\ref{thm:robust} bounding the trajectory deviation between systems $\stD$ and $\sthD$ is proved. Then, Corollary~\ref{cor:robustStable} that considers the special case when both $\stD$ and $\sthD$ are input-state stable is also proved.

%%%%%%%%%%%%%%%%%%%%%%%%%%%%%%%%%%%%%%%%%%%%%%%%%%%%%%%%%%%%%%%%%%%%%%%%%%%%%%%%
%%%%                      PROOF OF ROBUSTNESS THEOREM                       %%%%
%%%%%%%%%%%%%%%%%%%%%%%%%%%%%%%%%%%%%%%%%%%%%%%%%%%%%%%%%%%%%%%%%%%%%%%%%%%%%%%%

\begin{proof}[Proof of Theorem~\ref{thm:robust}]
The dynamic of the error graph signal $\mtX(t) - \mthX(t)$ is given by
\begin{equation} \label{eq:errorDynamics}
\begin{aligned}
    \mtX(t) - \mthX(t) = & \mtA \mtX(t-1) \mtbA - \mthA \mthX(t-1) \mthbA \\
    & + \mtB \mtU(t-1) \mtbB - \mthB \mthU(t-1) \mthbB.
\end{aligned}
\end{equation}
The evolution of $\mtX(t)$ and $\mthX(t)$ and that of $\mtU(t)$ and $\mthU(t)$ are studied separately.

To study the first part of the right-hand side of \eqref{eq:errorDynamics}, one can write:
% eq:Xsum3terms
\begin{align} \label{eq:Xsum3terms}
    & \mtA \mtX(t-1) \mtbA  - \mthA \mthX(t-1) \mthbA =  \mtA \mtX(t-1) \big( \mtbA - \mthbA \big)  \\
    &\quad + \big( \mtA - \mthA \big)  \mtX(t-1) \mthbA + \mthA \big( \mtX(t-1) - \mthX(t-1) \big) \mthbA .\nonumber
\end{align}
Observe that \eqref{eq:Xsum3terms} consists of three terms containing each of the errors between system matrices and states. Computing the size of the graph signal in \eqref{eq:Xsum3terms}, see \eqref{eq:graphSignalNorm}, and applying the triangular inequality for each of the three terms, one obtains:
%%
%\begin{equation}
%%
%\begin{aligned}
%%
%\big\|& \mtA \mtX(t-1) \mtbA  - \mthA \mthX(t-1) \mthbA \big\| \\
%    &\leq \sum_{g=1}^{F} \Big\| \sum_{f=1}^{F} \big([\mtbA]_{fg} - [\mthbA]_{fg} \big) \mtA \vcx^{f}(t-1)\Big\|_{2}   \\
%    & \quad + \sum_{g=1}^{F} \Big\| \sum_{f=1}^{F}[\mthbA]_{fg} \big(\mtA -\mthA \big)\vcx^{f}(t-1) \Big\|_{2}  \\
%    & \quad + \sum_{g=1}^{F} \Big\| \sum_{f=1}^{F} [\mthbA]_{fg}\mthA \big( \vcx^{f}(t-1)-\vchx^{f}(t-1) \big) \Big\|_{2}.
%%
%\end{aligned}
%%
%\end{equation}
%%
%We apply the triangular inequality again, with respect to the sum over $f$, and use the submultiplicativity of the Euclidean norm with respect to the spectral norm to obtain
%%
\begin{equation}
    \begin{aligned}
        \big\| & \mtA \mtX(t-1) \mtbA  - \mthA \mthX(t-1) \mthbA \big\| \\
        & \leq \|\mtA\|_{2} \sum_{f=1}^{F} \big\| \vcx^{f}(t-1)\big\|_{2} \sum_{g=1}^{F} \big|[\mtbA]_{fg} - [\mthbA]_{fg} \big|\\
        & \ +\big\|\mtA -\mthA \big\|_{2} \sum_{f=1}^{F}    \big\| \vcx^{f}(t-1) \big\|_{2} \sum_{g=1}^{F} \big| [\mthbA]_{fg} \big|\\
        & \ + \|\mthA\|_{2} \sum_{f=1}^{F} \Big\| \vcx^{f}(t-1)-\vchx^{f}(t-1) \Big\|_{2} \sum_{g=1}^{F} \big|[\mthbA]_{fg}\big| .
    \end{aligned}
\end{equation}
Now, using the bound $\sum_{g=1}^{F} | [\mthbA]_{fg} | \leq \max_{f} \sum_{g=1}^{F} | [\mthbA]_{fg}| = \| \mthbA\|_{\infty}$, and analogously for $(\mtbA-\mthbA)$, one can write:
% eq:boundErrorX
\begin{align} \label{eq:boundErrorX}
    & \big\|  \mtA \mtX(t-1) \mtbA  - \mthA \mthX(t-1) \mthbA \big\| \nonumber \\
    &\leq \Big( \|\mtA\|_{2} \|\mtbA - \mthbA \|_{\infty} + \|\mtA -\mthA \|_{2} \|\mthbA\|_{\infty} \Big)  \|\mtX(t-1)\| \nonumber  \\
    &\quad + \|\mthA\|_{2} \|\mthbA\|_{\infty} \|\mtX(t-1) - \mthX(t-1)\| ,
\end{align}
where the resulting sum over $f$ has been replaced for the corresponding size of the graph signal, see \eqref{eq:graphSignalNorm}.

Proceed analogously to \eqref{eq:boundErrorX}, the second term in the right-hand side of \eqref{eq:errorDynamics} can be bounded as
% eq:boundErrorU
\begin{align}\label{eq:boundErrorU}
    &\big\| \mtB \mtU(t-1) \mtbB  - \mthB \mthU(t-1) \mthbB \big\| \nonumber \\
    & \leq \Big( \|\mtB\|_{2} \|\mtbB - \mthbB \|_{\infty} + \|\mtB -\mthB \|_{2} \|\mthbB\|_{\infty} \Big)  \|\mtU(t-1)\| \nonumber \\
    & \quad + \|\mthB\|_{2} \|\mthbB\|_{\infty} \|\mtU(t-1) - \mthU(t-1)\|.
\end{align}
The control term $\|\mtU(t)\|$ is a GNN with input $\mtX(t)$ and can thus be bounded by leveraging Lemma~\ref{lemma:boundGCNN}, i.e. $\|\mtU(t) \| \leq C_{\fnPhi}\|\mtX(t)\|$. To bound $\|\mtU(t) - \mthU(t)\|$, $\fnPhi(\mtX(t);\mthS,\stH)$ is added and subtracted, and the size of the graph signal computed, to obtain
% eq:boundErrorUsplit
\begin{align} \label{eq:boundErrorUsplit}
    \big\| \mtU(t) - \mthU(t) \big\| = &\ \big\|\fnPhi\big(\mtX(t); \mtS, \stH\big) -\fnPhi\big(\mthX(t); \mthS, \stH\big) \big\| \nonumber \\
    &\ \leq \big\|\fnPhi\big(\mtX(t); \mtS, \stH\big) - \fnPhi\big(\mtX(t); \mthS, \stH\big)\big\|\\
    &\ \quad + \big\| \fnPhi\big(\mtX(t); \mthS, \stH\big) -\fnPhi\big(\mthX(t); \mthS, \stH\big) \big\| ,\nonumber
\end{align}
where the triangular inequality was used. For the first term in \eqref{eq:boundErrorUsplit}, it follows from Lemma~\ref{lemma:lipschitzGCNN} that:
% eq:boundUS
\begin{equation} \label{eq:boundUS}
    \big\| \fnPhi(\mtX(t);\mtS, \stH) - \fnPhi(\mtX(t);\mthS, \stH)\big\|  \leq \scGamma(\sceps) \scGamma_{\fnPhi} \| \mtX(t) \|,
\end{equation}
where $\scGamma(\sceps) = (1+8\sqrt{N})\sceps$ with $\sceps = \fnd(\stD, \sthD)$ depends on the characteristics of the support matrices $\mtS$ and $\mthS$, and where $\scGamma_{\fnPhi} = C_{\fnPhi} \sum_{\ell=1}^{L} \scGamma_{\fnH_{\ell}}/C_{\fnH_{\ell}}$ depends on the learned filters $\fnH_{\ell}(\cdot;\cdot,\stH)$. To bound the second term in \eqref{eq:boundErrorUsplit}, recall that the output of a GNN is its value at the last layer
% eq:boundErrorUsplit2ndTerm
\begin{equation} \label{eq:boundErrorUsplit2ndTerm}
\begin{aligned}
    \big\| \fnPhi\big(\mtX(t); \mthS, \stH\big) & -\fnPhi\big(\mthX(t); \mthS, \stH\big) \big\|  = \big\| \mtX_{L} - \mtX_{L} \big\| \\
    & = \big\| \fnsigma \big( \fnH_{L}(\mtX_{L-1};\mthS,\stH) \big) - \fnsigma \big( \fnH_{L}(\mthX_{L-1};\mthS,\stH) \big) \big\|.
\end{aligned}
\end{equation}
Using the assumption that $|\fnsigma(x) - \fnsigma(y)| \leq |x-y|$ for all $x,y \in \fdR$, \eqref{eq:boundErrorUsplit2ndTerm} can be upper bounded by
% eq:boundErrorUsplit2ndTermNoSigma
\begin{equation} \label{eq:boundErrorUsplit2ndTermNoSigma}
    \big\| \fnPhi\big(\mtX(t); \mthS, \stH\big) -\fnPhi\big(\mthX(t); \mthS, \stH\big) \big\| \leq  \big\| \fnH_{L}(\mtX_{L-1}-\mthX_{L-1};\mthS,\stH) \big\|.
\end{equation}
where the linearity of the filter with respect to the input $\mtX_{L-1}$ was used. Leveraging Lemma~\ref{lemma:boundFilter} on the upper bound of a graph filter, one obtains:
% eq:boundErrorUsplit2ndTermNextLayer
\begin{equation} \label{eq:boundErrorUsplit2ndTermNextLayer}
    \big\| \fnH_{\ell}(\mtX_{L-1}-\mthX_{L-1};\mthS,\stH) \big\| \leq C_{\fnH_{L}} \big\| \mtX_{L-1} - \mthX_{L-1} \big\|.
\end{equation}
Repeatedly applying \eqref{eq:boundErrorUsplit2ndTermNoSigma} and \eqref{eq:boundErrorUsplit2ndTermNextLayer}, the following upper bound on the second term of \eqref{eq:boundErrorUsplit} is obtained:
% eq:boundUX
\begin{equation} \label{eq:boundUX}
    \big\| \fnPhi\big(\mtX(t); \mthS, \stH\big) -\fnPhi\big(\mthX(t); \mthS, \stH\big) \big\| \leq \Big( \prod_{\ell=1}^{L} C_{\fnH_{\ell}} \Big) \big\| \mtX_{0} - \mthX_{0} \big\| = C_{\fnPhi} \big\| \mtX(t) - \mthX(t) \big\|,
\end{equation}
where the fact that the input to the GNN is the state at time $t$, i.e. $\mtX_{0} = \mtX(t)$. Finally, using \eqref{eq:boundUS} and \eqref{eq:boundUX} in \eqref{eq:boundErrorUsplit}, one obtains:
\begin{equation}
    \big\|\mtU(t) - \mthU(t)\big\| \leq  \scGamma(\sceps) \scGamma_{\fnPhi} \big\| \mtX(t)\big\| + C_{\fnPhi} \| \mtX(t) - \mthX(t) \|. \nonumber
\end{equation}
This simplifies \eqref{eq:boundErrorU} as
% eq:boundErrorUfull
\begin{align} \label{eq:boundErrorUfull}
    &\big\| \mtB \mtU(t-1) \mtbB  - \mthB \mthU(t-1) \mthbB \big\|  \\
    &\leq \Big(\|\mtB\|_{2} \|\mtbB - \mthbB \|_{\infty} + \|\mtB -\mthB \|_{2} \|\mthbB\|_{\infty} \Big)  C_{\fnPhi}\|\mtX(t-1)\| \nonumber \\
    &\quad + \|\mthB\|_{2} \|\mthbB\|_{\infty} \scGamma(\sceps) \scGamma_{\fnPhi} \big\| \mtX(t-1)\big\| \nonumber \\
    &\quad + \|\mthB\|_{2} \|\mthbB\|_{\infty} C_{\fnPhi} \| \mtX(t-1) - \mthX(t-1) \|. \nonumber
\end{align}

Now, computing the size of the error signal in \eqref{eq:errorDynamics} and using the triangular inequality, together with \eqref{eq:boundErrorX} and \eqref{eq:boundErrorUfull}, one obtains:
% eq:boundXt
\begin{align} \label{eq:boundXt}
    \big\| & \mtX(t) - \mthX(t) \big\|\leq \Big( \|\mthA\|_{2} \|\mthbA\|_{\infty} + \|\mthB\|_{2} \|\mthbB\|_{\infty} C_{\fnPhi} \Big) \| \mtX(t-1) - \mthX(t-1) \| \nonumber \\
    &  + \Big( \big( \|\mtA\|_{2} \|\mtbA - \mthbA \|_{\infty} + \|\mtA -\mthA \|_{2} \|\mthbA\|_{\infty} \big)  \\
    & \qquad \qquad + C_{\fnPhi} \big(\|\mtB\|_{2} \|\mtbB - \mthbB \|_{\infty} + \|\mtB -\mthB \|_{2} \|\mthbB\|_{\infty} \big) \Big)  \|\mtX(t-1)\| \nonumber  \\
    &+ \|\mthB\|_{2} \|\mthbB\|_{\infty} \scGamma(\sceps) \scGamma_{\fnPhi} \big\| \mtX(t-1)\big\|. \nonumber
\end{align}
Recall that $\schxi = \|\mthA\|_{2} \|\mthbA\|_{\infty} + C_{\fnPhi} \|\mthB\|_{2} \|\mthbB\|_{\infty}$ and note that
% eq:horribleCxi
\begin{equation} \label{eq:horribleCxi}
\big( \|\mtA\|_{2} \|\mtbA - \mthbA \|_{\infty} + \|\mtA -\mthA \|_{2} \|\mthbA\|_{\infty} \big) + C_{\fnPhi} \big(\|\mtB\|_{2} \|\mtbB - \mthbB \|_{\infty} + \|\mtB -\mthB \|_{2} \|\mthbB\|_{\infty} \big) \leq \schC_{\scxi} \sceps
\end{equation}
for $\schC_{\scxi}$ as in \eqref{eq:Cxi}. The value of $\|\mtX(t-1)\|$ can be further bounded as
\begin{equation}
    \|\mtX(t-1)\| \leq \big( \| \mtA \|_{2} \|\mtbA\|_{\infty} + C_{\fnPhi} \| \mtB \|_{2} \| \mtbB\|_{\infty} \big) \| \mtX(t-2) \|.
\end{equation}
Repeatedly applying this inequality, and noting that $\scxi = \| \mtA \|_{2} \|\mtbA\|_{\infty} + C_{\fnPhi} \| \mtB \|_{2} \| \mtbB\|_{\infty}$, see \eqref{eq:stabilityConstant}, the bound on $\|\mtX(t-1)\|$ becomes
% eq:boundX0tm1
\begin{equation} \label{eq:boundX0tm1}
    \|\mtX(t-1)\| \leq \scxi^{t-1} \| \mtX(0) \|.
\end{equation}
Using \eqref{eq:horribleCxi} and \eqref{eq:boundX0tm1} back in \eqref{eq:boundXt}, one obtains:
\begin{equation}
\big\|  \mtX(t) - \mthX(t) \big\|\leq \schxi\ \| \mtX(t-1) - \mthX(t-1) \|  + \big(  \schC_{\scxi} \sceps + C_{\fnPhi} \| \mthB\|_{2} \|\mthbB\|_{\infty} \scGamma(\sceps) \scGamma_{\fnPhi} \big)  \|\mtX(0)\| \scxi^{t-1},
\end{equation}
which can be conveniently rewritten as
% eq:boundErrorIneq
\begin{equation} \label{eq:boundErrorIneq}
e_{t} \leq \schxi e_{t-1} +  b\sceps \scxi^{t-1},
\end{equation}
with
\begin{IEEEeqnarray}{CL} \IEEEyesnumber \label{eq:boundErrorCoeff}
    e_{t} & = \| \mtX(t) - \mthX(t)\|, \IEEEyessubnumber \label{subeq:boundErrorCoeffE} \\
    \schxi &= \|\mthA\|_{2} \|\mthbA\|_{\infty}+C_{\fnPhi}^{L} \|\mthB\|_{2} \|\mthbB\|_{\infty} , \IEEEyessubnumber \label{subeq:boundErrorCoeffAhat} \\
    \scxi &= \| \mtA\|_{2} \| \mtbA\|_{\infty} + C_{\fnPhi}^{L} \| \mtB\|_{2} \| \mtbB\|_{\infty} ,  \IEEEyessubnumber \label{subeq:boundErrorCoeffA} \\
    b & = \big( \schC_{\scxi} + C_{\fnPhi} \| \mthB\|_{2} \|\mthbB\|_{\infty} (1+8\sqrt{N})\scGamma_{\fnPhi} \big) \| \mtX(0) \| ,\IEEEyessubnumber \label{subeq:boundErrorCoeffB}
\end{IEEEeqnarray}
where the definition of $\scGamma(\sceps) = (1+8\sqrt{N}) \sceps$ was used to highlight the linearity with $\sceps$. By repeatedly applying \eqref{eq:boundErrorIneq}, one arrives at:
% eq:boundErrorParams
\begin{equation} \label{eq:boundErrorParams}
e_{t} \leq b\sceps \sum_{\tau = 0}^{t-1} \scxi^{t-\tau-1} \schxi^{\tau} + \scxi^{t} e_{0}.
\end{equation}
Since the initial state of both the true system and the estimated one is the same, it holds that $e_{0} = \| \mtX(0) - \mthX(0)\| = 0$. Then, \eqref{eq:boundErrorParams} becomes
% eq:boundErrorParamsNoSum
\begin{equation} \label{eq:boundErrorParamsNoSum}
e_{t} \leq b\sceps \sum_{\tau = 0}^{t-1} \scxi^{t-\tau-1} \schxi^{\tau} =
\begin{cases}
b\frac{\scxi^{t} - \schxi^{t}}{\scxi - \schxi} & \text{ if } \scxi \neq \schxi \\
bt\scxi^{t-1} & \text{ if } \scxi = \schxi
\end{cases}.
\end{equation}
Now, recall that $|\scxi^{t} - \schxi^{t}| \leq t \max \{\scxi, \schxi\}^{t} | \scxi - \schxi|$ so that \eqref{eq:boundErrorParamsNoSum} becomes $e_{t} \leq bt \max\{\scxi, \schxi\}^{t-1} \sceps$. Finally, substituting the definitions of $e_{t}$ as in \eqref{subeq:boundErrorCoeffE}, $\schxi$ as in \eqref{subeq:boundErrorCoeffAhat}, $\scxi$ as in \eqref{subeq:boundErrorCoeffA}, and $b$ as in \eqref{subeq:boundErrorCoeffB}, we complete the proof.

\end{proof}

%%%%%%%%%%%%%%%%%%%%%%%%%%%%%%%%%%%%%%%%%%%%%%%%%%%%%%%%%%%%%%%%%%%%%%%%%%%%%%%%
%%%%                           PROOF OF COROLLARY                           %%%%
%%%%%%%%%%%%%%%%%%%%%%%%%%%%%%%%%%%%%%%%%%%%%%%%%%%%%%%%%%%%%%%%%%%%%%%%%%%%%%%%

Now we prove Corollary~\ref{cor:robustStable} for the particular case when both systems $\stD$ and $\sthD$ are input-state stable.

\begin{proof}[Proof of Corollary~\ref{cor:robustStable}]
From \eqref{eq:Ct} in Theorem~\ref{thm:robust} it holds that $\schC_{t} = t \max\{\scxi, \schxi\}^{t-1}$. By assumption, it is known that $\scxi<1$ and $\schxi < 1$. Therefore, the function $t \max\{\scxi, \schxi\}^{t-1}$ has a global maximum for $t \geq 0$. As a function of continuous $t \in \fdR$, this maximum is at $t = -1/\log(\max\{\scxi, \schxi\})$ and gives the optimal value $-e^{-1}/(\max\{\scxi, \schxi\} \times \log(\max\{\scxi, \schxi\}))$. Thus, it holds that $\schC_{t} \leq -e^{-1} \schC_{\fnPhi}/(\max\{\scxi, \schxi\} \times \log(\max\{\scxi, \schxi\}))$, completing the first part of the proof. For the second part, note that, since $\scxi < 1$ and $\schxi < 1$, then it holds that $\lim_{t \to \infty} t \max\{\scxi, \schxi\}^{t-1} = 0$.
\end{proof}

%%%%%%%%%%%%%%%%%%%%%%%%%%%%%%%%%%%%%%%%%%%%%%%%%%%%%%%%%%%%%%%%%%%%%%%%%%%%%%%%
%%%%                                                                        %%%%
%%%%                               REFERENCES                               %%%%
%%%%                                                                        %%%%
%%%%%%%%%%%%%%%%%%%%%%%%%%%%%%%%%%%%%%%%%%%%%%%%%%%%%%%%%%%%%%%%%%%%%%%%%%%%%%%%



\bibliographystyle{bibFiles/IEEEtranD}

%\ifundefined{arXiv}
%    \nocite{*}
%
%    \section*{\refname}
%\else
%\fi

\bibliography{bibFiles/myIEEEabrv,bibFiles/biblioDistributedLQR}

\end{document}